\documentclass[english, a4paper]{llncs}
\usepackage[T1]{fontenc}
\usepackage{babel}

\usepackage{float}

\DeclareMathAlphabet\gothic{U}{euf}{m}{n}

\usepackage{amsfonts}
\usepackage{amsmath, amscd, graphicx,amssymb}

\newcommand{\M}{\mathbb M}
\newcommand{\MM}{\mathcal M}
\newcommand{\G}{\mathcal G}
\newcommand{\spann}{\operatorname{span}\nolimits}
\newcommand{\x}{\mathbf x}
\newcommand{\RRS}{\mathbb{R}^2\!\rtimes\! \operatorname{SO}_2}
\usepackage{color}

\DeclareMathAlphabet\gothic{U}{euf}{m}{n}

\graphicspath{{figures/}}

\def\th{\theta}

\renewcommand{\Re}{\operatorname{Re}\nolimits}
\renewcommand{\Im}{\operatorname{Im}\nolimits}

 \newcommand{\argmax}{\operatorname{argmax}\nolimits}

\def\R{{\mathbb R}}

\def\N{{\mathbb N}}

\newcommand{\SE}{\operatorname{SE}_2}
\newcommand{\SOtwo}{\operatorname{SO}_2}

\newcommand{\e}{\epsilon}

\begin{document}

\title{Geometrical Optical Illusion via Sub-Riemannian Geodesics in the Roto-Translation Group}

\author{B.~Franceschiello\inst{1} \and A.~Mashtakov\inst{2} \and G.~Citti \inst{3} \and A.~Sarti \inst{4}}
\authorrunning{Franceschiello, Mashtakov, Citti, Sarti --------\today--------}
\institute{
$^1$ Fondation Asile des Aveugles and Laboratory for Investigative Neurophisiology, Department of Radiology, CHUV - UNIL, Lausanne\\
$^2$ Program Systems Institute of RAS, Russia, CPRC, $\left. \right.~~~~~~\,$
\\$^3$ Department of Mathematics, University of Bologna, Italy $\left. \right.$\\
$^4$ CAMS, Center of Mathematics, CNRS - EHESS,Paris, France
 \\ \email{benedetta.franceschiello@fa2.ch, alexey.mashtakov@gmail.com,\\ giovanna.citti@unibo.it, alessandro.sarti@ehess.fr} 
}

\maketitle

\begin{abstract}
We present a neuro-mathematical model for geometrical optical illusions (GOIs), a class of illusory phenomena that consists in a mismatch of geometrical properties of the visual stimulus and its associated percept. They take place in the visual areas V1/V2 whose functional architecture have been modelled in previous works by Citti and Sarti as a Lie group equipped with a sub-Riemannian (SR) metric. Here we extend their model proposing that the metric responsible for the cortical connectivity is modulated by the modelled neuro-physiological response of simple cells to the visual stimulus, hence providing a more biologically plausible model that takes into account a presence of visual stimulus. Illusory contours in our model are described as geodesics in the new metric. The model is confirmed by numerical simulations, where we compute the geodesics via SR-Fast Marching.
\end{abstract}

\section{Introduction}\label{sec:Intro}
Geometrical-optical illusions (GOIs) have been discovered in the XIX century by German psychologists (Oppel 1854~\cite{oppel1855uber}, Hering, 1878,~\!\cite{Her_1}) and have been defined as situations in which there is an awareness of a mismatch of geometrical properties between an item in the object space and its associated percept~\!\cite{westheimer2008illusions}. These illusions induce a misjudgement of the geometrical properties of the stimulus, due to the perceptual difference between the features of the presented stimulus and its associated perceptual representation. An historical survey of the discovery of geometrical-optical illusions is included in Appendix I of~\!\cite{westheimer2008illusions} and a classification of these phenomena, can be found in Coren and Girgus, 1978,~\!\cite{coren1978seeing}; Robinson, 1998,~\!\cite{robinson2013psychology}; Wade, 1982,~\!\cite{wade1982art}.

The aim of this paper is to propose a mathematical model for GOIs based on the functional architecture of low level visual cortex (V1/V2). This neuro-mathematical model will allow to interpret at a neural level the origin of GOIs and to reproduce the arised percept for this class of phenomena. The main idea is to adapt the model for the functional geometry of V1 provided in~\!\cite{CS1} for perceptual completion. Here we extend it introducing a new metric for the connectivity of the visual cortex, wich takes into account the output of simple cells in V1/V2, as a coefficient modulating the sub-Riemannian metric. We also postulate that geometrical optical illusory curves arise as geodesics in this new connectivity metric between two given sets. Then we will adapt to this definition the SR Fast-Marching (SR-FM) algorithm introduced in~\!\cite{duits2018,sanguinetti2015sub} as tool for the computation of geodesics with fixed two-point boundary conditions (extrema points). As a result we will be able to explain the perceptual phenomena by the geometry of V1
and SR differential geometry instruments. 

The paper is organised as follows. In Section~\ref{sec:1}, we introduce the perceptual problem of geometrical optical illusion and we review the state of the art concerning the previously proposed mathematical models. In Section~\ref{sec:prel}, we briefly recall the functional architecture of the visual cortex and the cortical based model introduced by Citti and Sarti in~\cite{CS1}. In Section~\ref{sec:4}, we introduce the neuro-mathematical model proposed for GOIs, taking into account the modulation of the functional architecture induced by the stimuli. 
In Section~\ref{sec:5}, we discuss sub-Riemannian geodesics and the sub-Riemannian Fast-Marching. 
Finally, we describe the implementation in Section~\ref{sec:modeling}  and discuss the results in Section~\ref{sec:Experiments}.
\section{Geometrical optical illusions}
\label{sec:1}
\subsection{Hering, Ehm--Wackermann and Poggendorff illusions}
\label{sec:11}

\begin{figure}[htbp]
\centering
\includegraphics[width=\textwidth]{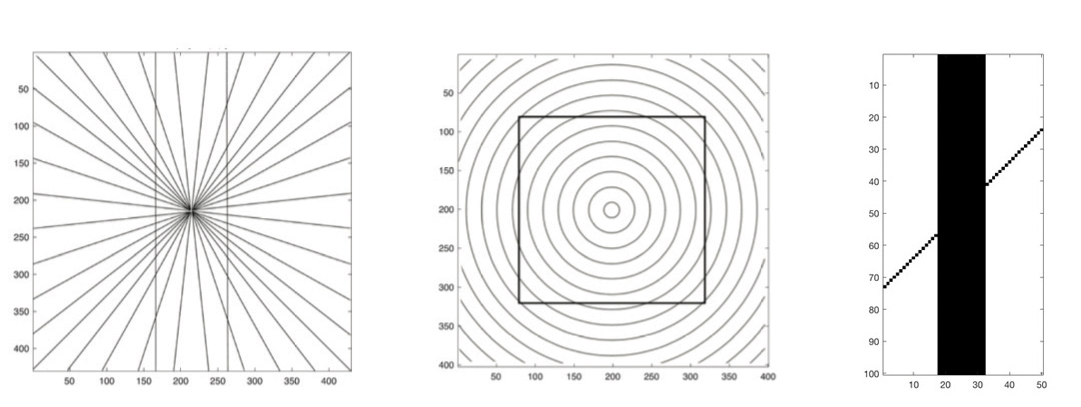} 
 
 \caption{From left to right. Hering illusion: two straight vertical lines in front of a radial background appear as if they were bowed outwards. Ehm--Wackermann illusion: the context of concentric circles bends inwards the edges of the square. Poggendorff illusion: the presence of a central surface induces a misalignment of the crossing transversals.}
 		 \label{fig:1ills}
\end{figure}

The phenomena we consider here consist in misperception effects induced by elements of the image. The Hering illusion, introduced by Hering in 1861~\!\cite{Her_1}, is presented in Figure~\ref{fig:1ills}, left. In this illusion two vertical straight lines are presented in front of a radial background, so that the lines appear as if they were bowed outwards. A similar effect is observable in the Ehm--Wackermann illusion~\!\cite{ehm2012modeling}, i.e. a square on a background composed by concentric circles, Figure~\ref{fig:1ills}, center. One more famous GOI is the Poggendorff illusion, which consists in an apparent misalignment of two collinear, oblique, transversals separated by a rectangular surface (Figure~\ref{fig:1ills}, right). For the latter, psychological elements contributing to this misperception have been presented in~\!\cite{day1976components,robinson2013psychology} and~\!\cite{talasli2015applying}. 

The interest in GOI comes from the chance to provide a better explanation of these phenomena, helping to understand the unrevealed mechanisms of vision~(\cite{eagleman2001visual}). 
Many studies, which rely on neuro-physiological and imaging data, show the evidence that neurons in at least two visual areas, V1 and V2, carry signals related to illusory contours, and those signals in V2 are more robust than in V1~(\cite{von1984illusory},~\!\cite{murray2002spatiotemporal}, reviews \cite{eagleman2001visual},~\!\cite{murray2013illusory}). A more recent study measured the activated connectivity in and between areas of early visual cortices~(\cite{song2013effective}). To integrate the mathematical model with the recent findings, we propose a neural-based model to interpret GOI.

\subsection{Mathematical models proposed in literature}
\label{sec:12}
The first models of GOI are purely phenomenological and provide \textit{quantitative} analysis of the perceived geometrical distortion, such as the angle deformation, which is the attitude of perceiving acute angles larger and obtuse ones as smaller. Models of this type have been proposed in 1971 by Hoffmann in terms of orbits of a Lie group acting on the plane~\!\cite{hoffman1971visual}, and by Smith~\!\cite{smith1978descriptive} in terms of differential equations. More recently Ehm and Wackermann in~\!\cite{ehm2012modeling} and~\!\cite{ehm2016geometric} provided a variational approach expressed by a functional dependent on length of the curve and the deflection from orthogonality along the curve. These approaches do not take into account the underlying neurophysiological mechanisms. 

On the other hand an entire branch for modelling neural activity, the Bayesian framework, had its basis in Helmholtz theory \cite{von2005treatise}: \textit{our percepts are our best guess as to what is in the world, given both sensory data and prior experience.} 
The described idea of unconscious inference is at the basis of the Bayesian statistical decision theory, a principled method for determining optimal performance in a given perceptual task~(\cite{geisler2002illusions}). An application of this theory to motion illusions has been provided by Weiss et al in~\!\cite{weiss2002motion}, by Geisler and Kersten in~\!\cite{geisler2002illusions}, by Ferm\"{u}ller and Malm in~\!\cite{Ferm}. 
In our model, we aim to combine psycho-physical evidence and neurophysiological findings, in order to provide a neuro-mathematical description of GOIs. It is inspired by the celebrated models of Hoffman~\!\cite{hoffman1989visual} and Petitot \cite{petitot1999vers,petitot2008neurogeometrie}, who have founded a discipline now called neuro-geometry, aimed to describe the functional architecture of the visual cortex with geometrical instruments in order to explain phenomenological evidence. More recent contributions are due to August and Zucker~\!\cite{august2000curve}, Sarti and Citti~\!\cite{CS1,sarti2008symplectic}, Duits et al.~\!\cite{duits2010left,duits2010left2}. A recent work trying to integrate the neuro-physiology of V1/V2 for explaining such phenomena has been presented in~\!\cite{franceschiello2017neuro}.

\section{The classical neuromathematical model of V1/V2}\label{sec:prel}

\subsection{The set of simple cells receptive profiles}

The retina, identified as $\M \subset \mathbb{R}^2$, is the first part of the visual path initiating the signal transmission, which passes through the Lateral Geniculate Nucleus and arrives in the visual cortex, where it is processed. 

Let us consider a visual stimulus, i.e. an image
\begin{equation}\label{eq:intensity}
I: \M \subset \mathbb{R}^2 \rightarrow \mathbb{R}^{+}.
\end{equation}
The receptive field (RF) of a cortical neuron is the portion of the retina which the neuron reacts to, and the receptive profile (RP) $\psi(\chi)$ is the function that models its activation when a point $\chi = (\chi_1,\chi_2) \in \R^2$ of the retinal plane is elicited by a stimulus at that point. To be specific, $(\chi_1,\chi_2)$ are the local coordinates of the neighbourhood centered at $\x = (x,y) \in \R^2$, to which the neuron reacts to, while $(x,y)$ refers to the global coordinates system of the retina $\mathbb{R}^2$. Simple cells of V1 are sensitive to position and orientation of the contrast gradient of an image~$\nabla I $. Their properties have been discovered by Hubel and Wiesel in~\!\cite{hubel1962receptive} and experimentally described by De Angelis in~\!\cite{deangelis1995receptive}. 
Considering a basic geometric model, the set of simple cells RPs can be obtained via translation on the vector $\x$ and rotation on the angle $\theta \in \mathit{S}^1 \simeq \SOtwo$ of a unique mother profile $\psi_0(\chi)$.

\subsection{Receptive profiles and Gabor filters}

Receptive fields have been modelled as oriented filters in the middle of 80's and since then the orientation extraction in image analysis has been subject of several works. The first models have been presented by Daugman~\!\cite{Daug}~\!(1985), Jones and Palmer~\!\cite{jones1987evaluation}~\!(1987) in terms of Gabor filters. In the same years Young in~\!\cite{young1987gaussian}~\!(1987) and Koenderink in~\!\cite{koenderink1990receptive}~\!(1990) proposed to model RPs as Gaussian derivatives~\!(DoG). We also refer to \cite{Bart}~\!(2008) and~\!\cite{petitot2008neurogeometrie}~\!(2008) for further explanations and details. Recently a new class of multi-orientation filters have been introduced by Duits et al. in~\!\cite{duits2007image}~\!(2007): cake-wavelets. A comparison between cake-wavelets and Gabor filters efficiency has been presented in~\!\cite{bekkers2014}. Having the scope of modelling the functionality of the visual cortex, we chose Gabor filters, proved to be a good model of receptive profiles and their spiking responses~\cite{PetkovCyber}. We will consider odd and even part of Gabor filters in order to measure $\theta$ correctly for both contours and lines. 

\begin{definition}
A mother Gabor filter is given by
 \begin{equation}
\psi_0 (\chi) =  \frac{\alpha}{2 \pi \sigma^2} e^{\frac{-(\chi_1^2 + \alpha^2 \chi_2^2)}{2\sigma^2}}e^{\frac{2  i \chi_2}{\lambda}}, \qquad \chi = (\chi_1, \chi_2)\in \R^2,
 \label{receptive_profile}
 \end{equation}
 where $\lambda>0$ is the spatial wavelength, $\alpha>0$ is the spatial aspect ratio and $\sigma>0$ is the standard deviation of the Gaussian envelope.
\end{definition}

As discovered by Lee~\!\cite{lee1996image}, the set of simple cells RPs can be obtained via translation on the vector $(x,y) \in \R^2$ and rotation on the angle $\theta \in \mathit{S}^1$ of a unique mother profile $\psi_0(\chi)$. 
Since the set of parameters $(x,y,\theta)$ describes the Lie group $\SE$ of rotations and translations, we identify the set of receptive profiles (RPs) with this group. Let $\eta = (x,y,\theta) \in \SE$. Then an action of the Lie group $\SE$ on the homogeneous space $\R^2$ is given by
\begin{equation} \label{group_law}
\eta \odot \chi = \left(\begin{array}{l}
  x \\ y 
\end{array}\right) + \left(\begin{array}{ll}
 \cos\theta &-\sin\theta \\ \sin\theta &\cos\theta 
\end{array}\right) \left(\begin{array}{l}
 \chi_1 \\ \chi_2
\end{array} \right).                 
\end{equation}
 Denote $\eta^{-1} = (-x \cos\th - y \sin \th, x \sin\th - y \cos\th, -\th) \in \SE$ the inverse element to $\eta$. A general RP can be expressed as
\begin{equation} \label{eq:geRP}
\psi_\eta(\chi)=\psi_0(\eta^{-1}\odot\chi).
\end{equation}

\subsection{Output of receptive profiles and Gabor transform}\label{sec:coticmodel}

The retina is the light-sensitive layer of shell tissue of the eye, where the visual stimulus is first imprinted. In our model, we follow~\cite{CS1} and identify the retina with the planar area $\M \subset \R^2$. 

A visual stimulus is modelled by intensity function~(\ref{eq:intensity}):
$$I: \M \subset \mathbb{R}^2 \rightarrow \mathbb{R}^+: (x,y) \mapsto I(x,y).$$ 
It activates the retinal layer of photoreceptors. Then the cortical neurons whose receptive fields intersect the activated layer spike. We model their spiking responses (or spiking activity)~$O_I(\eta)=O_I(x,y,\theta) $ by a Gabor transform.

\begin{definition}
Let $\psi_\eta \in L^2(\R^2)$ be a Gabor filter, given by \eqref{eq:geRP}, modelling the receptive profile of simple cells of the primary visual cortex. The continuous Gabor transform of a signal $I\in L^2(\R^2)$ is defined as:
\begin{equation}\label{eq:genericGaborTransformDefn}
 O_I(\eta)=\int_{\R^2} 
I(\chi) \psi_\eta(\chi)\, d\chi.
 \end{equation} 
\end{definition}
In the previous definition the output $O_I(\eta)$ of receptive profiles of simple cells in V1 in response to a visual stimulus $I(x,y)$ is mathematically described as a convolution. 
Let us note that Gabor filters are complex valued: the real and imaginary parts have a different role and detect different features. The real part is even and spikes maximally along lines, while the imaginary part is odd and detects the presence of surfaces, i.e. contours. 

\subsection{Hypercolumnar structure}
	\label{sec:221}
	The term \textit{functional architecture} refers to the organisation of cells of the primary visual cortex in structures. 
	The hypercolumnar structure, discovered by the neuro-physiologists Hubel and Wiesel in the 60s~(\cite{hubel1977ferrier}), organizes the cells of V1/V2 in columns (called hypercolumns) covering a small part of the visual field $\mathbb{R}^2$ and corresponding to parameters such as orientation, scale, direction of movement, color, for a fixed retinal position $(x,y)$, Figure \ref{hyper} (top). 
	\begin{figure}[H] 
		\centering
		\includegraphics[width=0.5\columnwidth]{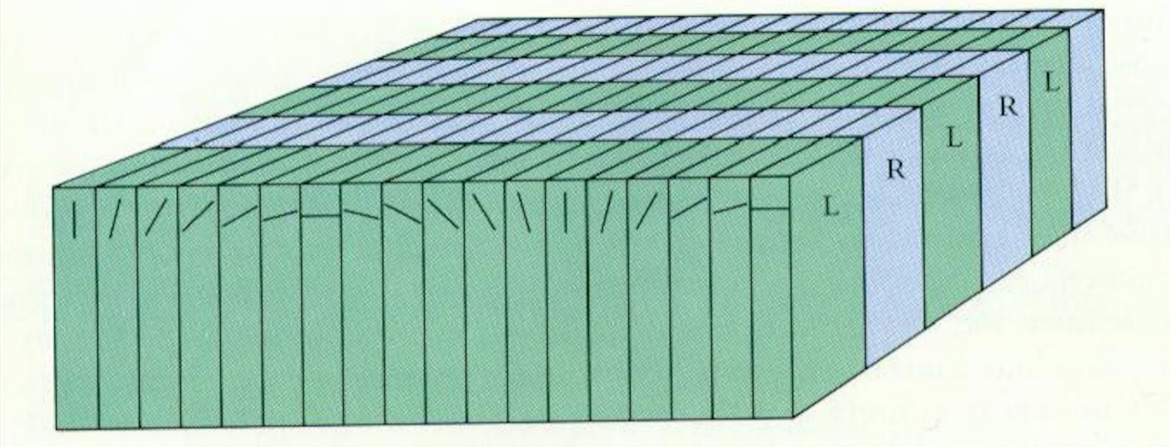}   \includegraphics[width=0.5\columnwidth]{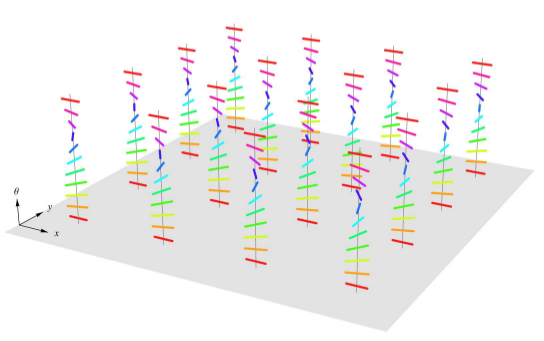}
		\caption{ Top: representation of hypercolumnar structure, for the orientation parameter, where L and R represent the ocular dominance columns (Petitot \cite{petitot2008neurogeometrie}). Bottom: the set of all possible orientations for each position of the retina $(x,y)$.}\label{hyper}
	\end{figure}
	In our framework, over each retinal point we consider a whole hypercolumn of cells, sensitive to all possible orientations, see Figure~\ref{hyper} (bottom). Hence for each position $(x,y)$ of the retina $\M \subset \mathbb{R}^2$ we associate a whole set of filters 
	\begin{equation}
	RP_{(x,y)} = \{ \psi_{(x,y,\theta)}:  \,\, \theta \in \mathit{S}^1 \} . 
	\end{equation}
	This expression 
	defines a fiber $\{\theta \in \mathit{S}^1 \}$  over each point $(x,y) \in \R^2$. 
	
	
	In this framework the hypercolumnar structure is described in terms of differential geometry, but further explanations are requested to model the orientation selectivity process performed by the cortical areas in the space of feature~\!$\mathit{S}^1$~\!(\cite{CS1}).

\subsection{Cortical connectivity and sub-Riemannian structure}\label{subsec:Statement}
\label{sec:223}
From the physiological point of view the orientation selectivity is the action of short range connections between simple cells belonging to the same hypercolumn to select the most probable instance from the spiking activation of receptive profiles in response to a stimulus. 

Mathematically, this process is modelled by assignment to every point $\x = (x,y) \in \R^2$ the angle $\bar{\theta} \in \mathit{S}^1$ --- the orientation of a line passing through the point $\x$. It is found as the element of fiber that gives the maximal response of~(\ref{eq:genericGaborTransformDefn}):
\begin{equation}
\bar{\theta}(x,y) = \argmax\limits_{\theta \in \mathit{S}^1}\,\, | O_I (x,y,\theta) |.
\end{equation}
This process is called \textit{lifting} and it associates to each retinal point $(x,y)$ the corresponding maximal output $\bar{\theta}(x,y)$, denoting the selected orientation (tangent direction) to the visual stimulus at point $\x$. 

The other connectivity, responsible for the formation of contours in the cortex by given a retinal stimulus, is called horizontal connectivity. Horizontal connections are long ranged and connect cells of approximately \textit{the same orientation}, belonging to different hypercolumns. Modelling this behaviour requires to endow V1 with a differential structure, see~\cite{CS1}, where horizontal curves are the lifting of retinal curves to the extended space of positions and orientations --- the Lie group $\SE \cong \RRS$.
The horizontal connectivity is therefore modelled as a diffusion along the integral curves of the left invariant vector fields on the group. 

The basis of left-invariant vector fields on $\SE$ is given by
\begin{equation*}
X_1 = \cos \theta \frac{\partial}{\partial x} + \sin \theta \frac{\partial}{\partial y}, \quad 
X_2 = \partial_\theta,\quad  X_3= -\sin \theta \frac{\partial}{\partial x} + \cos\theta \frac{\partial}{\partial y}.
\end{equation*}
In order to model the propagation of the horizontal connectivity in $\RRS$, in~\cite{CS1} Citti and Sarti proposed to endow $\RRS$ with a sub-Riemannian metric.

\begin{definition} A SR manifold is given by a triple 
$(M, \Delta, \mathcal{G}),$ where $M$ is a connected, simply connected smooth manifold, 
$\Delta$ is a smooth subbundle of the tangent bundle to $M$, and $\mathcal{G}$ is a metric defined on $\Delta$. 
\end{definition} 

In particular, in~\cite{CS1}, the horizontal connectivity in the cortex is modelled by means of distance function, defined on the SR manifold $(M, \Delta, \G)$, where 
\begin{equation}\label{H-invese_metric_pesata}
M=\SE,\quad \Delta = \spann(X_1,X_2), \quad \G =  \omega^1 \otimes \omega^1 + \omega^2 \otimes \omega^2.
\end{equation}
Here $\otimes$ denotes the Kronecker product, and $\omega^i \in T^{*}M$ denotes the basis one form dual to $X_i$, i.e. $\big\langle \omega^i, X_j \big\rangle = \delta_{ij}$, where $\delta_{ij}$ is the Kronecker delta.


\section{The polarized connectivity metric of V1/V2} \label{sec:4}
In the previous section we provided neuro-geometrical tools that, starting from the neural counterpart of the visual cortex, explain the behaviour of V1/V2 in presence of visual stimuli. 
The original contribution of this paper is to extend the previous model~\cite{CS1} in the following way: starting from the sub-Riemannian metric $\G$ in \eqref{H-invese_metric_pesata}
, we weight the long range connectivity taking into account the intra-columnar response of simple cells in V1/V2. 

We will obtain a new polarized metric $\mathcal{G}_0$ in \eqref{X-invese_metric_pesata}. Illusory countours in our model arise as local minimizers (i.e. geodesics) of this polarized metric. 

\subsection{Polarization of the metric}\label{sec:metric_polarization} 
\begin{figure}[ht]
\centering
\includegraphics[width=0.275\linewidth]{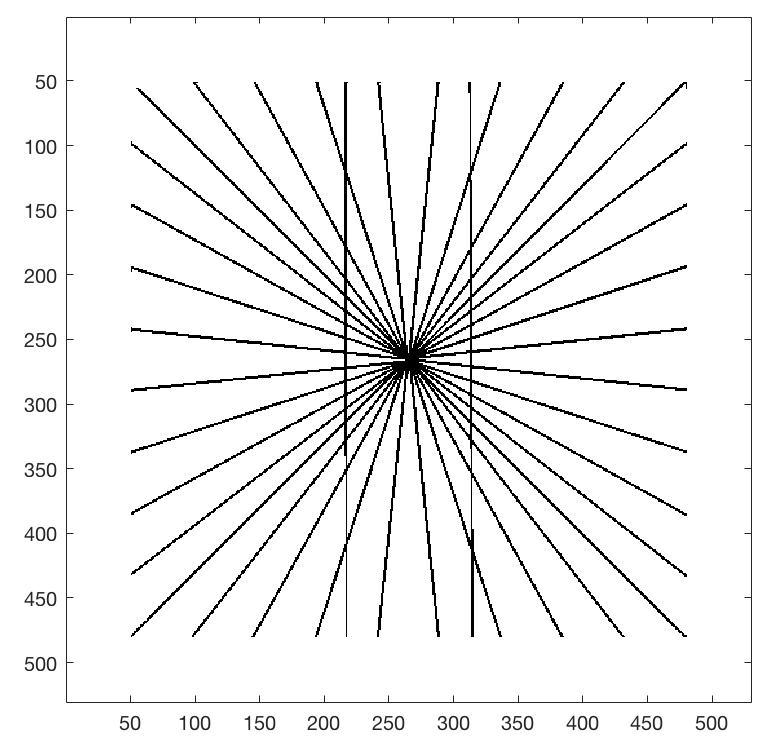} \qquad
\includegraphics[width=0.32\linewidth]{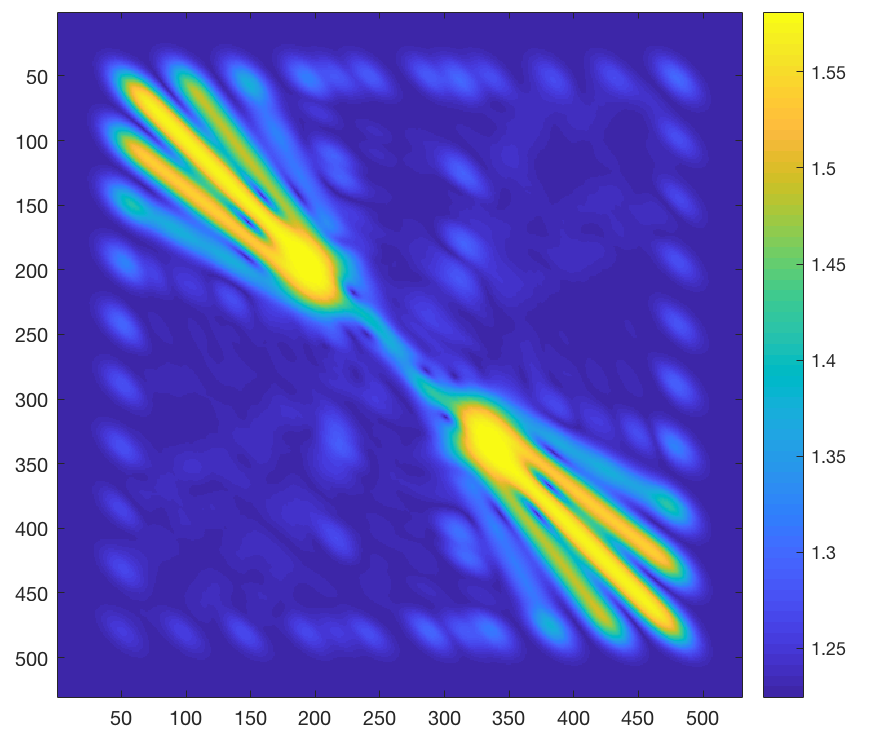} \\
\includegraphics[width=0.275\linewidth]{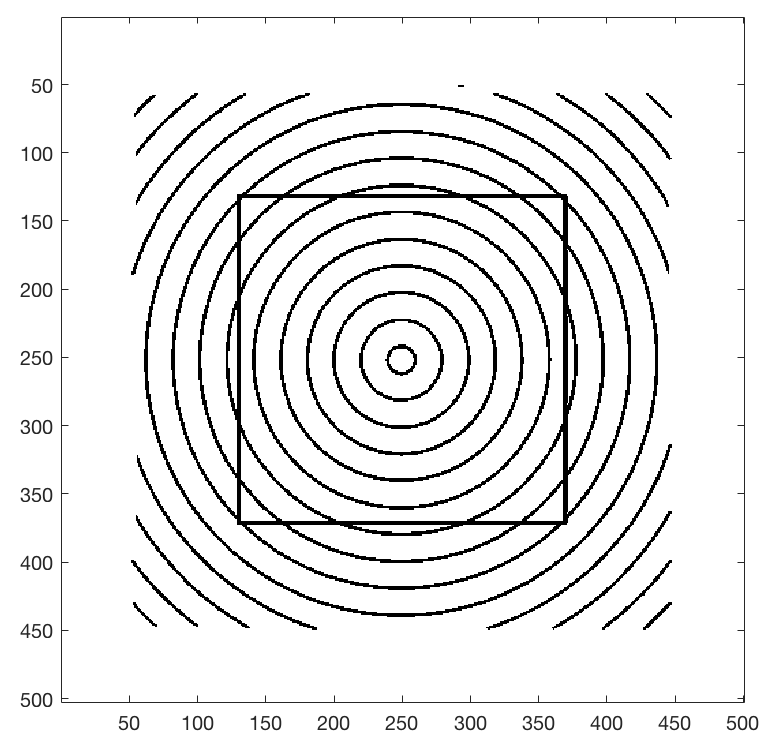}\qquad
\includegraphics[width=0.32\linewidth]{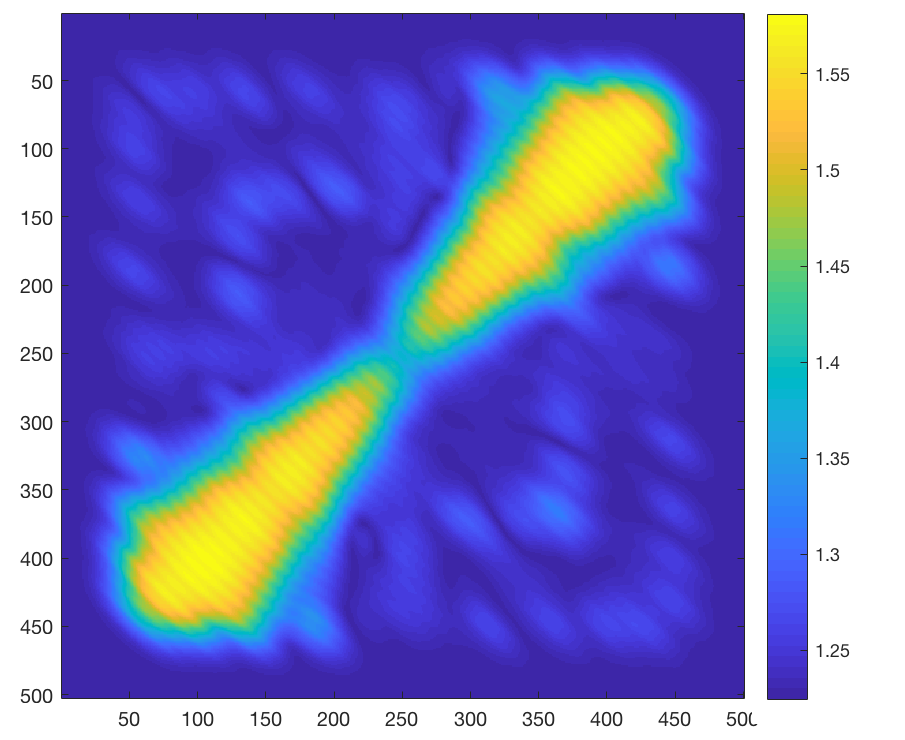}
\caption{{\bf Top Left:} Hering illusion, cf. Fig.~\ref{fig:1ills}.
 {\bf Top Right:} Level line of the external cost $\mathcal{C}(x,y,\theta)$ for $\theta = 2.3167$ rad. {\bf Bottom Left:} Ehm--Wackermann illusion, cf. Fig.~\ref{fig:1ills}.
{\bf Bottom Right:} Level line of the external cost $\mathcal{C}(x,y,\theta)$ for $\theta = 2.3167$ rad.} \label{costo_her}
\end{figure}
\begin{figure}[ht]
\centering
\includegraphics[width=0.24\linewidth]{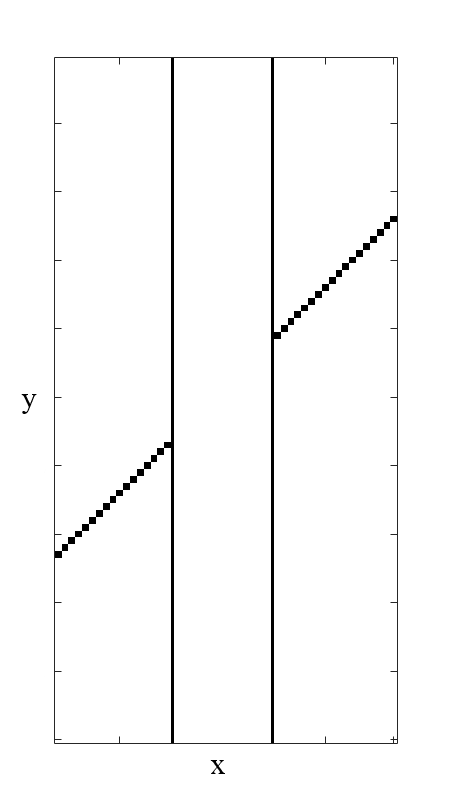} \qquad
\includegraphics[width=0.24\linewidth]{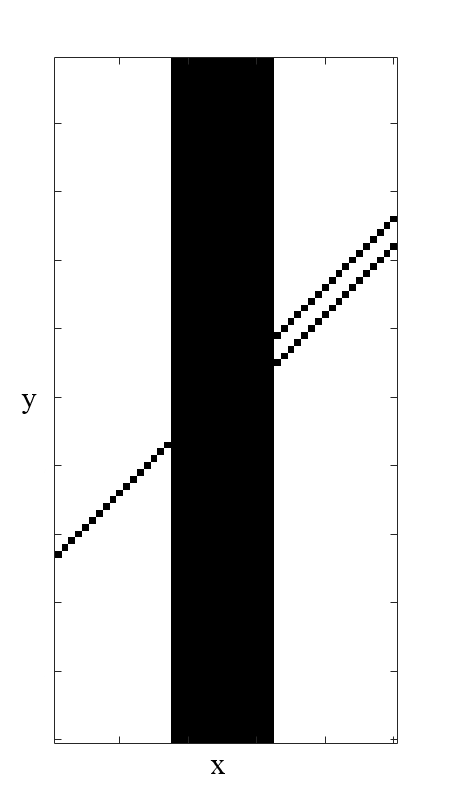} \qquad
\includegraphics[width=0.275\linewidth]{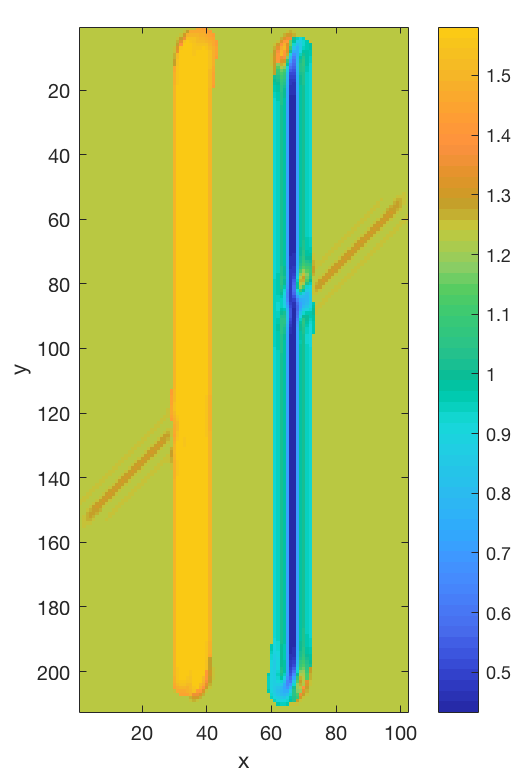} 
\caption{{\bf Left:} Poggendorff illusion, cf. Fig.~\ref{fig:1ills}. 
 {\bf Center:} the initial stimulus with a second transversal corresponding to the perceptual completion. {\bf Right:} 
A level set of 
$\mathcal{C}(x,y,\theta)$ 
for $\theta = 2.83$ rad. The saturation is slightly lowered 
to show detected contours and~lines.
} \label{pogg_main_exp}
\end{figure}

Here we introduce the idea that the isotropic cortical metric $\G_0$ defined on the horizontal subbundle $\Delta$, see \eqref{H-invese_metric_pesata}, can be modulated by the output of simple cells $O_I(\eta)$, induced by the visual stimulus $I$. 
This phenomena is a weak type of learning, or pre-activation, where the activated cells are more sensitive to the cortical propagation. The proposed modulation $P(O_I)$ of the metric induced by the visual stimulus $I$
is maximal in correspondence of the edges, and is expressed~as
\begin{equation}
P(O_I) = P(\eta) =  \Re(O_I(\eta))^2 + \Im(O_I(\eta)), \label{P}
\end{equation}
where $O_I(\eta$) is the output of simple cells as defined in~\eqref{eq:genericGaborTransformDefn}. 

This formulation allows to detect both the presence of lines (first term) and the presence and polarity of contours (second term).
Once computed $O_I(\eta$) from the initial images, we add to $P(\eta)$ a positive values as the values of the output $O_I$ range from negative to positive value. The modification will make $P(\eta)$  a positive value, which can be considered a polarization term for the metric. Finally, we normalize it, obtaining the following external cost:
\begin{equation}
\mathcal{C}(O_I) =\mathcal{C}(\eta) = \frac{c + P(\eta)}{\sqrt{c +{P(\eta)}^2}},
\label{external_cost}
\end{equation}
where $c$ is a suitable positive constant.

We define the polarized metric on the distribution $\Delta = \spann(X_1, X_2)$ as 
\begin{equation} \label{X-invese_metric_pesata} 
	\mathcal{G}_0= \,\, \left(\begin{array}{cc}
			      \frac{1}{ \mathcal{C}_\xi(O_I)}  & 0 \\
			      0 &  \frac{1}{ \mathcal{C}(O_I)}  
			      \end{array}\right),
			      \quad
			      \mathcal{G}_0^{-1} = \left(\begin{array}{cc}
			      \mathcal{C}_\xi(O_I)  & 0 \\
			      0 &  \mathcal{C}(O_I)   
			      \end{array}\right), 
\end{equation}
where $\mathcal{C}_{\xi} (O_I) = \xi^2 \mathcal{C}(O_I)$. Here $\xi>0$ is a real parameter, which will be fixed and kept constant. It allows to weight differently the translational, i.e. $X_1$, $X_3$, and rotational, i.e.  $X_2$ components of the metric. In other words it allows to modulate the anisotropy between the retinical (on $\mathbb{R}^2$) and the hypercolumnar ($\mathit{S}^1$) components. The choice of the constant $\xi$ is discussed across experiments in Section~\ref{sec:modeling}. Let us notice we denoted the metric as $\mathcal{G}_0$, where 
coefficient $0$ indicates that the metric has no contribution in the direction $X_3$. 
 
\subsection{Sub-Riemannian metric}
 A curve $\gamma: [0,T] \mapsto \SE$ is an integral curve of a vector field $X$ starting from the point $a$ iff $\dot{\gamma}(t) = X|_{\gamma(t)} $ and $\gamma(0)= a$. In this last case $\gamma$ will be 
also denoted $$\exp(t X)(a) = \gamma(t).$$

Denote $\MM = (\SE, \Delta, \mathcal{G}_0)$ the SR manifold with polarized metric \eqref{X-invese_metric_pesata}.

\begin{definition}
A Lipschitzian curve $\gamma:[0,T]\to \SE$ on $\MM$ is called horizontal, if $\dot{\gamma}(t) \in \Delta|_{\gamma(t)}$ for a.e. $t \in [0,T]$.
\end{definition} 
In other words, a horizontal curve $\gamma$ on $\MM$ is an integral curve of
$$\dot{\gamma}(t) = u_1(t) X_1|_{\gamma(t)} + u_2(t) X_2|_{\gamma(t)},$$
where $u_1$, $u_2$ are real-valued functions from $L_{\infty}([0,T]).$

\begin{definition} 
The SR length of a horizontal curve $\gamma$ on $\MM$ is defined as
\begin{equation}\label{eq:geodcontsystorig}
\begin{array}{c}
l(\gamma) := \int\limits_0^{T} \sqrt{\mathcal{G}_{0}(\dot{\gamma}(t), \dot{\gamma}(t))}\, {\rm d} t. 
\end{array}
\end{equation}
\end{definition}
For a general introduction to sub-Riemannian geometry see~\cite{montgomery}. Note, that the vector fields $X_1$, $X_2$ satisfy the H\"ormander condition~\cite{montgomery}, i.e. their generated Lie algebra coincides with the tangent space at every point. Due to this property, the Rashevski (1938)--Chow (1939) theorem guarantees  that any two elements $\eta_0$, $\eta_1 \in \SE$ can be connected by a horizontal curve. 



Hence in a connected manifold $\MM$ 
for any couple of points $\eta_0$ and $\eta_1$ the
following set is not empty:
$$\Gamma_{\eta_0, \eta_1} =\{ 
\gamma \text{ horizontal curve}, \gamma(0) = \eta_0, \gamma(T) = \eta_1 \}.
 $$
As a consequence it is possible to define a distance on a connected manifold $\MM$.

\begin{definition} 
The Carnot-Carath\'eodory distance on the sub-Riemannian manifold $\MM$  between two points $\eta_0$ and $\eta_1$ is defined as
\begin{equation}
d_0(\eta_0,\eta_1)= \underset{
\gamma \in \Gamma_{\eta_0, \eta_1}
}{\inf}  l(\gamma). \label{def_geod}
\end{equation}
\end{definition}

Note that  Filippov's theorem~\cite{agrachev2013control} implies existence of length-minimizers and infimum in~(\ref{def_geod}) can be replaced by minimum.

\subsection{Riemannian approximation}
Computation of sub-Riemannian (Carnot--Carath\'eodory) distance in general is a very difficult problem. For example, even in the case of left-invariant SR structures on Lie groups the length-minimizers are known only in several simplest cases: the Heisenberg group~\cite{Brockett1982,Vershik1988}, the groups $\operatorname{SO}_3$, $\operatorname{SU}_2$, and $\operatorname{SL}_2$ with the Killing metric~\cite{Boscain2008}, $\SE$~\cite{sachkovSE2fin}, $\operatorname{SH}_2$~\cite{Butt2017}, the Engel group~\cite{Ardentov2013}, and 2-step corank 2 nilpotent SR problems~\cite{Barilari2011}. Our case $\MM = (\SE, \Delta, \mathcal{G}_0)$ is much more difficult than the case of a left-invariant SR metric, since the metric $\mathcal{G}_0$ depends on the functional parameter -- the external cost $\mathcal{C}$. Thus, to obtain an analytic expression for SR distance~(\ref{def_geod}) does not seem possible. 

Instead, based on idea of Riemannian approximation~\cite{sanguinetti2015sub}, we build a numerical method to compute the SR-distance as a limiting case of the Riemannian distances, when the penalization of movement in the direction $X_3$ (forbidden in SR case) tends to infinity. In other words, the Riemannian approximation relaxes the horizontality constraint $\dot{\gamma}\in\Delta$ and extends the SR metric $\mathcal{G}_0$ to the highly anisotropic Riemannian metric $\mathcal{G}$ defined in whole tangent bundle $T\SE$.  

There are several possible definitions for Riemannian distance
functions which approximate a Carnot-Carath\'eodory distance in the Gromov-Hausdorff sense. We use the following Riemannian approximation for the SR metric~$\mathcal{G}_0$. 
\begin{definition}[Riemannian approximation of the SR metric]
A Riemannian approximation $\mathcal{G}_{\epsilon}$ of the sub-Riemannian metric $\mathcal{G}_0$ is defined over the whole tangent bundle $T\SE$ and has the following expression in the frame $(X_1,X_2,X_3)$:
$$\mathcal{G}_{\epsilon}= \operatorname{diag}\Big(\frac{1}{\mathcal{C}_{\xi}}, \frac{1}{\mathcal{C}},  \frac{1}{\epsilon^2 \mathcal{C}_{\xi}} \Big),$$
where $\xi>0$, $\epsilon>0$ are parameters of the metric anisotropy, $\mathcal{C}= \mathcal{C}(O_I)$ is external cost~(\ref{external_cost}) and $\mathcal{C}_{\xi}= \xi^2 \mathcal{C}$. 
\end{definition}

\begin{remark}
Note that the metric blows up as $\epsilon$ tends to $0$. On the other hand the inverse of the metric, computed as $$\mathcal{G}^{-1}_{\epsilon} =   \operatorname{diag}\Big(\mathcal{C}_{\xi}, \mathcal{C},  \epsilon^2 \mathcal{C}_{\xi} \Big)$$
formally tends to $\mathcal{G}^{-1}_{0}$. 
\end{remark}
For every $\epsilon>0$ we denote as $d_\epsilon(\cdot, \cdot)$ the Riemannian distance associated to the metric $\mathcal{G}_{\epsilon}$. 
The following result, proved in \cite[Theorem 1.1]{ge1993collapsing} in general settings, 
provide a relationship between Riemannian and sub-Riemannian distance:
\begin{theorem} Let $M = \SE$. The sequence $(M, d_\epsilon)$ converges to $(M, d_0)$ as $\epsilon\to 0$ in the Gromov-Hausdorff sense. 
\end{theorem}
 See also \cite[Theorem 1.2.1]{monti-tesi} for another related Riemannian approximation scheme and \cite{duits2018} for an extension to more general Finsler metrics  by Duits et al. In order to better understand this assertion we provide explicit estimates of the approximated distance. To this end, let us recall that the exponential map is a local diffeomorphism around each point (see~\cite{nagel1985balls}).

\begin{proposition}[\cite{nagel1985balls}]
For every fixed point $\eta_0 \in \SE$ the function 
$$\Phi_{\eta_0}(\zeta) = \exp\left(\sum_{i=1}^3 \zeta_i X_i\right)(\eta_0),\qquad \zeta = (\zeta_1,\zeta_2,\zeta_3) \in \R^3,$$
is a local diffeomorphism around the point $\eta_0$ and 
its inverse: $\Phi_{\eta_0}^{-1}$ defines local coordinates in a neighborhood of $\eta_0$. 
\end{proposition}

To better describe the dependence of the distance 
$d_\epsilon$ on the parameter $\epsilon$, let us define the regularized gauge function:
\begin{equation*}\label{Nepsilon}
N_0(\zeta)=\sqrt{\zeta_1^2+\zeta_2^2+
 |\zeta_3|}, \quad 
N_\epsilon(\zeta)=\sqrt{\zeta_1^2+\zeta_2^2+
 \min\big( |\zeta_3|,  \epsilon^{-2}\zeta_3^2
 \big)}, \, \,\,\,\text{for}\,\epsilon >0
\end{equation*}
and the associated pseudo-distance function $$
	\gothic{d}_{\epsilon} (\eta_0,\eta_1)=N_\epsilon (\Phi_{\eta_0}^{-1}(\eta_1) )$$ which 
	provides an estimate of the distance $d_\epsilon$ for both $\epsilon=0$ and 
$\epsilon>0$ (see \cite{capogna2016regularity}).

\begin{lemma}[\cite{capogna2016regularity}] \label{N=d}
There exists $A>0$, independent of $\epsilon$, such that for all $\eta_0,\eta$
\begin{equation}\label{ball-box}
A^{-1} \;	 \gothic{d}_{\epsilon}(\eta ,\eta_0)   \le d_\epsilon (\eta ,\eta_0) \le A \;
  \gothic{d}_{\epsilon}(\eta ,\eta_0)  .
\end{equation}
\end{lemma}

The dependence from $\epsilon$ becomes now quite clear. 
Indeed for every fixed $\zeta$: 
$$\epsilon \to 0 \; \Rightarrow \; N_\e(\zeta) \to N_0(\zeta),$$
which is independent of $\epsilon$, and provides an estimate for the SR distance.

\section{GOIs as Sub-Riemannian Geodesics}\label{sec:5}

\subsection{Distance function from a set}

A second original aspect of our model is the way we recover 
subjective contours. Many results using sub-Riemannian (SR) geodesics as model for subjective contours in $\SE$ are already present in literature. 
SR geodesics and their application to image analysis were also studied in~\cite{BenYos,Hladky,MashtakovMTMA}, e.g. to retinal vessel tracking by Bekkers et al.~\cite{bekkers2015pde,Mashtakov2017SO3,PTR2}. Explicit formulas for SR geodesics and optimal synthesis in $\SE$ are obtained by Sachkov~\cite{sachkovSE2fin}.  Let us also underline that in~\cite{CS1} geodesics arise as foliation of subjective surfaces.

The problem we face here is more general, since we do not 
know the exact position of the geodesic extrema. 
Let us consider an example, the Hering illusion: it presents a misperception of two vertical straight lines, perceived 
as bowed outwards. The perceived curves are modelled as geodesics 
of the polarized metric in $\SE$, but we only know the spatial $\mathbb{R}^2$-components of its extrema ($(x_0, y_0)$ and $(x_1, y_1)$), while the angles $\theta_0$ and $\theta_1$ are  unknown. As a result the reconstructed perceptual curve is described as the minimizing geodesic between two a priori known sets, obtained by fixing the spatial component $(x,y) \in \R^2$ and varying the orientation  $\theta \in \mathit{S}^1$. 

We can generalize this problem by means of the following definition.

\begin{definition}\label{min_curve}
Let $K_0, K_1 \subset M = \SE$ be compact and non empty sets.\\
 For fixed $\epsilon \geq\! 0$, consider the (Riemannian, if $\epsilon >0$, or SR, if~$\epsilon =0$) metric $\mathcal{G}_\epsilon$.\\ We~call $\epsilon$-\-minimizing geodesic 
with extrema in the sets $K_0$ and $K_1$ the curve $\gamma_\epsilon$ (horizontal, if $\epsilon=0$), such that 
\begin{equation}\label{eq:epsmingeod}
l_{\epsilon} (\gamma_\epsilon) = \min\limits_{\tilde{\gamma}} \{l_{\epsilon}(\tilde \gamma)\; | \; \tilde \gamma : [0,T]\rightarrow M, \dot{\tilde \gamma}(t) \in \Delta_\epsilon,  \tilde \gamma (0) \in K_0, 
\tilde \gamma (T) \in K_1 \}.
\end{equation}
Here $l_{\epsilon}$ denotes the length in  $\mathcal{G}_\epsilon$, and $\Delta_\epsilon = \spann(X_1,X_2, \epsilon X_3) \subseteq TM$, $\epsilon \geq 0$.%
\end{definition}
Note that minimum in~(\ref{eq:epsmingeod}) exists due to compactness of the sets $K_0$ and $K_1$ and existence of a minimizing geodesic connecting any two given points $\eta_0 \in K_0$ and $\eta_1 \in K_1$, as we will formally show in Proposition~\ref{prop:exist}.

In other words, an $\epsilon$-minimizing geodesic realizes the distance $d_\epsilon$ between two sets $K_0$  and $K_1$, where the distance is defined in the following. 

\begin{definition} Let $K_0, K_1 \subset \SE$ be compact 
 non empty sets.\\ 
The distance function from $K_0$ (Riemannian if $\epsilon >0$ or SR if
$\epsilon =0$) is defined~as
$$d_{\epsilon, \, K_0}(\eta) = \inf
_{\eta_0\in K_0}
d_\epsilon (\eta, \eta_0).$$
Hence, the distance function between two sets $K_0, K_1$ is defined as
$$ d_\epsilon(K_0, K_1)  = \inf_{\eta_1 \in K_1} d_{\epsilon, \, K_0}(\eta_1). $$ \label{distance_def}
\end{definition}
Furthermore, if $\gamma$ is $\epsilon$-minimizing geodesic with extrema in $K_0$ and $K_1$, then
\[\gamma (0) \in K_0, \quad 
 \gamma (T) \in K_1, \quad l_\epsilon(\gamma) = d_\epsilon(K_0, K_1).\]
\begin{remark}\label{rem:exist}
The special case in which $K_0=\{\eta_0\}$ and $K_1=\{\eta_1\}$ 
contain only one point, the previous problem reduces to find the 
length minimizing curve between $\eta_0$ and $\eta_1$. 
The existence of a minimum is well known and it is called minimizing 
geodesic. We refer to~\cite{montgomery} for general properties of 
minimizing geodesics. 
\end{remark}
From the existence of a geodesic with fixed extrema, 
we can deduce the existence of $\epsilon$-minimizing geodesic in the sense of Definition~\ref{min_curve}. 

\begin{proposition}\label{prop:exist}
In the assumptions of Definition~\ref{min_curve} the minimum in~(\ref{eq:epsmingeod}) exists.
\end{proposition}
\begin{proof}
Indeed we can find a sequence $\eta_{0, n}$ in $K_0$ and a 
sequence $\eta_{1, n}$ in $K_1$ such that $d_\epsilon(\eta_{0, n}, \eta_{1, n})$ tends to $d_\epsilon(K_0, K_1)$. Since  $\eta_{0, n}$ and  $\eta_{1, n}$ 
are bounded they have a limit, respectively 
$\eta_{0}$ and  $\eta_{1}$. The geodesic curve between these two 
points exists by Remark~\ref{rem:exist} and gives the minimum in~(\ref{eq:epsmingeod}).
\end{proof}

From the convergence of the distance $d_\epsilon$ to the 
distance $d_0$ as $\epsilon \rightarrow 0$, we can deduce the following proposition:

\begin{proposition} 
Let $K_0, K_1 \subset \SE$ be compact non empty sets.\\ The following convergence result holds:
$$d_{\epsilon}(K_0, K_1)\rightarrow d_{0}(K_0, K_1); 
\quad d_{\epsilon, K_0}(\eta)\rightarrow d_{0, K_0}(\eta)
$$ 
as $\epsilon \rightarrow 0 $
\end{proposition}
\begin{proof}
By definition of distances $d_\epsilon$ and $d_0$ we immediately see 
that $\forall \epsilon >0$:
$$0 \leq d_\epsilon(K_0, K_1) \leq d_0(K_0, K_1).$$ 
For any sequence $\epsilon_j \to 0$ as $j\to\infty$ consider the sequence $d_{\epsilon_j}(K_0, K_1)$. Since $d_{\epsilon}(K_0, K_1) $ is bounded, $d_{\epsilon_j}(K_0, K_1)$ is 
converging. 
Clearly the limit $l$ satisfies $l \leq d_0(K_0, K_1)$. 
Let us assume by contradiction that 
$d_{\epsilon_j}(K_0, K_1)$ converges to a limit $l < d_0(K_0, K_1)$. 
For every $j$ there exist $\eta_{0, {\epsilon_j}}\in K_0$ and
$\eta_{1, {\epsilon_j}}\in K_1$ such that 
$$d_{\epsilon_j}( \eta_{0, \epsilon_j}, \eta_{1, \epsilon_j}) = 
d_{\epsilon_j}(K_0,  K_1).$$
Since $(\eta_{0, \epsilon_j})_{j \in \N}$ and 
$(\eta_{1, \epsilon_j})_{j \in\N}$ 
are bounded, then they have a convergent subsequences:
$\eta_{0, j}\to \eta_{0}$ and $\eta_{1, j}\to \eta_{1}$.  
Now we note that 
$$
|d_{\epsilon_j}(\eta_{0, \epsilon_j}, \eta_{1, \epsilon_j}) 
- d_{\epsilon_j}(\eta_{0, \epsilon_j}, \eta_{1}) 
|+| d_{\epsilon_j}(\eta_{0, \epsilon_j}, \eta_{1}) 
- d_{\epsilon_j}(\eta_{0}, \eta_{1}) | \to 0,
$$
since $d_{\epsilon_j}$ is Lipshitz continuous, with 
Lipshitz constant $1$. 
Hence 
$$d_{\epsilon_j}(\eta_{0, \epsilon_j}, \eta_{1, \epsilon_j}) \leq
|d_{\epsilon_j}(\eta_{0, \epsilon_j}, \eta_{1, \epsilon_j}) 
- d_{\epsilon_j}(\eta_{0, \epsilon_j}, \eta_{1}) 
|+$$$$+| d_{\epsilon_j}(\eta_{0, \epsilon_j}, \eta_{1}) 
- d_{\epsilon_j}(\eta_{0}, \eta_{1}) | + d_{\epsilon_j}(\eta_{0}, \eta_{1})\to 
 d_0
 (\eta_{0}, \eta_{1})$$
and as a result 
$$d(K_0, K_1) \leq  d_0
 (\eta_{0}, \eta_{1}) = \lim_{\epsilon_j \to 0} d_{\epsilon_j}(\eta_{0, \epsilon_j}, \eta_{1, \epsilon_j}), $$
and this is a contradiction, so that 
$d_{\epsilon_j}(K_0, K_1)\to d_{0}(K_0, K_1)$. 
Since any sequence has the same limit, then $d_{\epsilon}(K_0, K_1)\to d_{0}(K_1, K_2)$ as $\epsilon \to 0$. 
\end{proof}

\subsection{Riemannian and sub-Riemannian eikonal equation}

Now we show that the distance function from a set satisfies a first order PDE called eikonal equation. 
We first recall the notion of a (sub-)Riemannian gradient.

\begin{definition}
The Riemannian gradient of a function $f$ in
the metric $\mathcal{G}_{\epsilon}$ is the vector 
$$\nabla_{\epsilon} f = \sum_{j=1}^3\mathcal{G}_{\epsilon}^{ij} \left(X_j f\right) X_i.$$
For $\epsilon =0$ we obtain the sub-Riemannian gradient
$$\nabla_{0} f = \sum_{j=1}^2\mathcal{G}_{0}^{ij} \left(X_j f\right) X_i.$$
\end{definition}

In the Riemannian setting it is well know that the distance function from a set is a viscosity solution (in the sense by \cite{crandall1983viscosity,crandall1992user}) of the eikonal equation: 
\begin{proposition}[\cite{crandall1983viscosity,crandall1992user}]
Let $K$ be a compact non empty set with $\mathit{C}^\infty$ boundary and $\epsilon>0$ be a constant. Then in the points of differentiability outside the set $K$ the Riemannian distance function $d_{\epsilon, K}(\eta)$ satisfies the equation
$$\lVert \nabla_{\epsilon} d_{\epsilon, K}(\eta) \rVert_{\mathcal{G}_\e}=1,
$$
and $d_{\epsilon, K}(\eta)$ vanishes in $K$.
\end{proposition}
In Carnot groups with SR metric the same assertion has been proved in~\cite{monti-tesi}. 
This result can be extended to the present setting:
\begin{proposition}
The distance function $d_{0,K}$ given by Definition~\ref{distance_def} satisfies the following eikonal equation:
\begin{equation} \label{eq:eikonal}
\!
\left\{
\begin{array}{l}
\lVert \nabla_{0} d_{0,K}(\eta) \rVert_{\mathcal{G}_0}=1,
\textrm{ for }\eta\notin K, \\
d_{0,K}(\eta)=0, \quad \textrm{ for }\eta\in \partial K.
\end{array}
\right.
\end{equation}
\end{proposition}
We omit the proof which is similar to the one contained in \cite{monti-tesi}.
Viceversa the following result holds:
\begin{proposition}
The problem
\begin{equation} \label{problem:eikonal}
\!
\left\{
\begin{array}{l}
\lVert \nabla_{0} u \rVert_{\mathcal{G}_0}=1,
\textrm{ for }\eta\notin K, \\
u=0, \quad \textrm{ for }\eta\in \partial K.
\end{array}
\right.
\end{equation}
 has
a unique viscosity solution, $u$ which
coincides with the distance function $d_{0,K}(\eta)$ from the set $K$ in the 
sub-Riemannian setting (for $\epsilon =0$) or Riemannian setting (for $\epsilon >0$).
\end{proposition}
The proof can be obtained working as in \cite{barles1994solutions}

\subsection{Sub-Riemannian Fast Marching }

One of the most efficient method to compute geodesics in the Euclidean setting is Fast-Marching, introduced by Sethian in~\!\cite{Sethian99a}.  In case of geodesics between two points it has been extended by Mirebeau (in the case of Riemannian metric~\!\cite{mirebeau2014anisotropic}), and by Sanguinetti et al in~\!\cite{sanguinetti2015sub} (in the $\SE$ setting with a general SR metric). 

The Fast-marching method works as follows:
\begin{itemize}
\item First the distance map, from the initial point is computed as 
viscosity solution of the eikonal equation. 
\item Then a backtracking procedure is applied. The latter is based on the relationship between the gradient of the distance function and the direction of the geodesics far from cut points, i.e. points where 
the geodesics loses minimality (see~\cite{agrachev2013control}). 
\end{itemize}
In particular, in~\cite{bekkers2015pde} it was shown 
that if $\eta_0 \not= \eta\in \SE$ and the unique minimizing geodesic $\gamma_{\epsilon}:[0, T] \to \SE$  
joining $\eta_0$ and $\eta$ does not contain cut points, then  
$\dot{\gamma}(t) = \nabla_\epsilon d_\epsilon(\gamma(t), \eta_0)$.
As a consequence, they show that the 
geodesic can be recovered with the following backtracking procedure: 

\begin{proposition}\label{prop:backtrack}
By given two distinct points $\eta_0 \neq \eta \in \SE$ consider a geodesic $\gamma_{\epsilon}:[0, T] \to \SE$ joining $\eta_0$ and $\eta$. If $\gamma$  
 does not contain cut points 
then $\gamma_\epsilon (t)=\gamma_{\epsilon b}(T-t)$, 
where $\gamma_{b, \epsilon}$ is a solution of the problem
\begin{equation}\label{BT}
\left\{
\begin{array}{l}
\dot{\gamma}_{\epsilon b}(t) = - \nabla_\epsilon d_\epsilon(\gamma_{\epsilon b}(t), \eta_0), \ \ t \in [0,T], \\
\gamma_{\epsilon b}(0)=\eta.
\end{array}
\right.
\end{equation}
\end{proposition}

\subsection{SR Fast Marching for the distance function from a set}

Here we extend  backtraking procedure~(\ref{BT}) in Proposition~\ref{prop:backtrack} to the geodesics with extrema in a set, introduced in Definition~\ref{min_curve}. 
Before that we make a remark.

\begin{remark}
If $K_0$ and $K_1$ are compact sets with smooth boundary  and 
$d_{\epsilon, K_0}$ attains a minimum at the point $\eta_1$ 
on the set $K_1$, then the minimizing geodesic with extrema in  $K_0$ and $K_1$ coincides with the minimizing geodesic with first extremum in $K_0$ and second extremum $\eta_1$. 
\end{remark}

As a consequence, the minimizing geodesic can be found via 
backtracking of the distance from the starting set:

\begin{proposition}
Let $\epsilon > 0$. If the following assumptions are satisfied:
\begin{enumerate}
	\item $K_0, K_1 \subset \SE$ are compact non empty sets with smooth boundary;
	\item $d_{\epsilon, K_0}$ attains a minimum at the point $\eta_1$ 
on the set $K_1$;
\item minimizing geodesic $\gamma_{\epsilon}$  
joining $K_0$ and $\eta_1$ does not contain cut points;
\end{enumerate}
then $\gamma_\epsilon (t)=\gamma_{\epsilon b}(T-t)$,
where $\gamma_{\epsilon b}$ is a solution of the problem:
\begin{equation}\label{BT2}
\left\{
\begin{array}{l}
\dot{\gamma}_{\epsilon b}(t) = - \nabla_\epsilon d_{\epsilon, K_0}(\gamma_{\epsilon b}(t)), \ \ t \in [0,T] \\
\gamma_{\epsilon b}(0)=\eta_1.
\end{array}
\right.
\end{equation}
\end{proposition}
\begin{proof}
We can assume that $\gamma_\epsilon$ is parametrized by arclength. 
Then we have 
$$t = d_{\epsilon, K_0}(\gamma_\epsilon(t))$$
and differentiating with respect to $t$ we get
$$1 = \langle \nabla_ \epsilon d_{\epsilon, K_0}(\gamma_\epsilon(t)), 
\dot{\gamma}_\epsilon (t) \rangle_\epsilon \leq \|
\nabla_ \epsilon d_{\epsilon, K_0}(\gamma_\epsilon(t))\|\|
 \dot{\gamma}_\epsilon(t)\|\leq 1. $$
Since equality holds, then $\dot{\gamma}_\epsilon(t)$ is 
parallel to $\nabla_ \epsilon d_{\epsilon, K_0}(\gamma_\epsilon(t)).$
But they have the same norm, hence
$\dot{\gamma}_\epsilon(t) = \nabla_\epsilon d_{\epsilon, K_0}(\gamma_\epsilon(t)).$
Since $\gamma_{b, \epsilon}(t)= \gamma_\epsilon (T-t),$
then 
$$\dot{\gamma}_\epsilon(t) = - \nabla_\epsilon d_{\epsilon K_0}(\gamma_\epsilon(t)).$$
\end{proof}
In the limit for $\epsilon \to 0$ we recover a minimizing geodesic for the SR problem.
\begin{corollary} 
If $K_0$ and $K_1$ are compact non empty sets with smooth boundary, and 
for every $\epsilon >0$ the $\epsilon$-minimizing geodesic 
$\gamma_{\epsilon}:[0, T] \to \SE$  
joining $K_0$ and $K_1$ does not contain cut points, 
then there exists 
$$\lim_{\epsilon\to 0} \gamma_{\epsilon} = \gamma_0$$
and $\gamma_0$ is the minimizing SR geodesic 
joining $K_0$ and $K_1$. 
\end{corollary} 

Clearly the minimizing geodesic between two sets can be 
as well computed via minimization on the complete set 
of geodesics connecting points of $K_0$ and~$K_1$.
 
\section{Implementation}\label{sec:modeling}

\subsection{Choice of the Gabor filters}

The imaginary part of the output corresponds to the response of odd Gabor filters (contours polarity), while the real part performs the orientation detection of lines (line detection). We assume that the orientation domain $\theta$ takes values in $[-\pi, \pi)$. As $Re(O_I(x,y,\theta))$ has $\pi$-periodicity, the energy volume is duplicated to insure a well definition of the activation responses over the whole angular domain $[-\pi, \pi)$. Let us notice that while combining odd and even receptive profile, since odd receptive profiles over a straight line do not produce any response, we took different scales for even and odd filters: even Gabor profiles need to be sharp to detect line orientation, while odd Gabor profiles need to be wider to detect whether they are aligned or not along a surface contour.

\subsection{Discretization parameters}\label{subsec:exp0}
The first step performed consists into the convolution of the initial image with a bank of even and odd Gabor filters. A response $O_I$ is produced and opportunely combined to obtain $P$, as described in \eqref{P}. $P$, corresponding to the polarization of our SR metric, is  shifted to positive values and normalized to obtain $\mathcal{C}(x,y,\theta)$, finally used as weight for the connectivity (figure \ref{costo_her}, right column).
The SR geodesic that solves~(\ref{def_geod}) is obtained in two steps: 
\begin{enumerate}
\item Computation of the distance map solving~(\ref{eq:eikonal}) via Sub-Riemannian Fast-Marching, see figure \ref{Geod_45}, left;

\item Computation of the geodesic by gradient descent~(\ref{BT2}). 

\end{enumerate}
The constructed metric in $\RRS$ is a Riemannian approximation of the SR metric, weighted by the external cost $\mathcal{C}(x,y,\theta)$.
When switching from image coordinates to mathematical coordinates one should take care of correctly evaluating $\xi$, which represents the anisotropy between the two horizontal direction, $\xi \Delta x = \Delta \theta$, where $\Delta x, \Delta \theta$ are the discretization steps along $x$ and $\theta$. 
In the experiments for Hering and Ehm--Wackermann illusions, we set $\epsilon = 0.1$, $\xi = 7$, while it varied proportionally to the geometrical elements of the image (entry transversal and width of the surface) in the experiments for Poggedorff illusion. 
As was shown in~\cite{sanguinetti2015sub}, $\epsilon$ is taken sufficiently small to give an accurate approximation of the SR-case.

\section{Results}\label{sec:Experiments}
We processed the initial stimuli of the illusions through the method presented in Section~\ref{sec:metric_polarization} and implemented in Section~\ref{subsec:exp0}. 
\begin{figure}[htpb]
\centering
\includegraphics[width=0.24\linewidth]{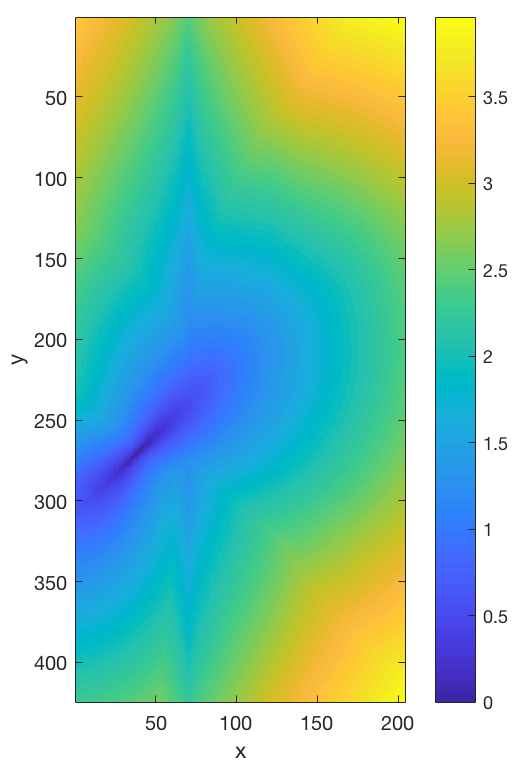}  \includegraphics[width=0.245\linewidth]{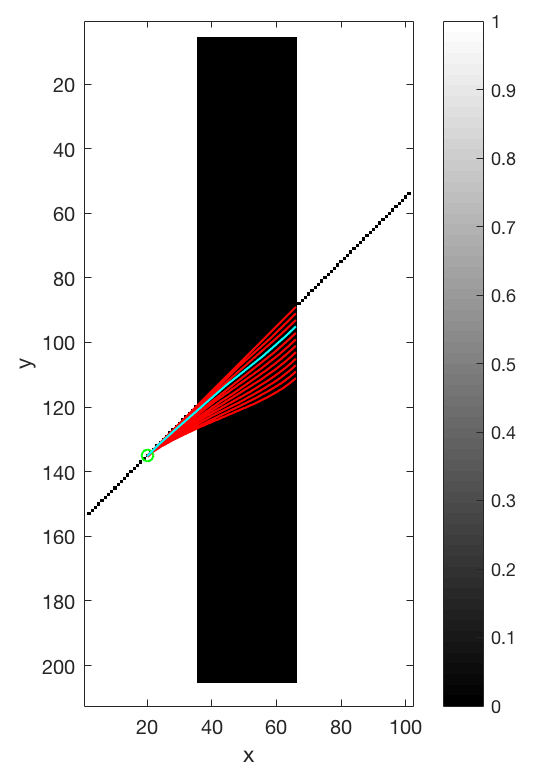} 
\includegraphics[width=0.495\linewidth]{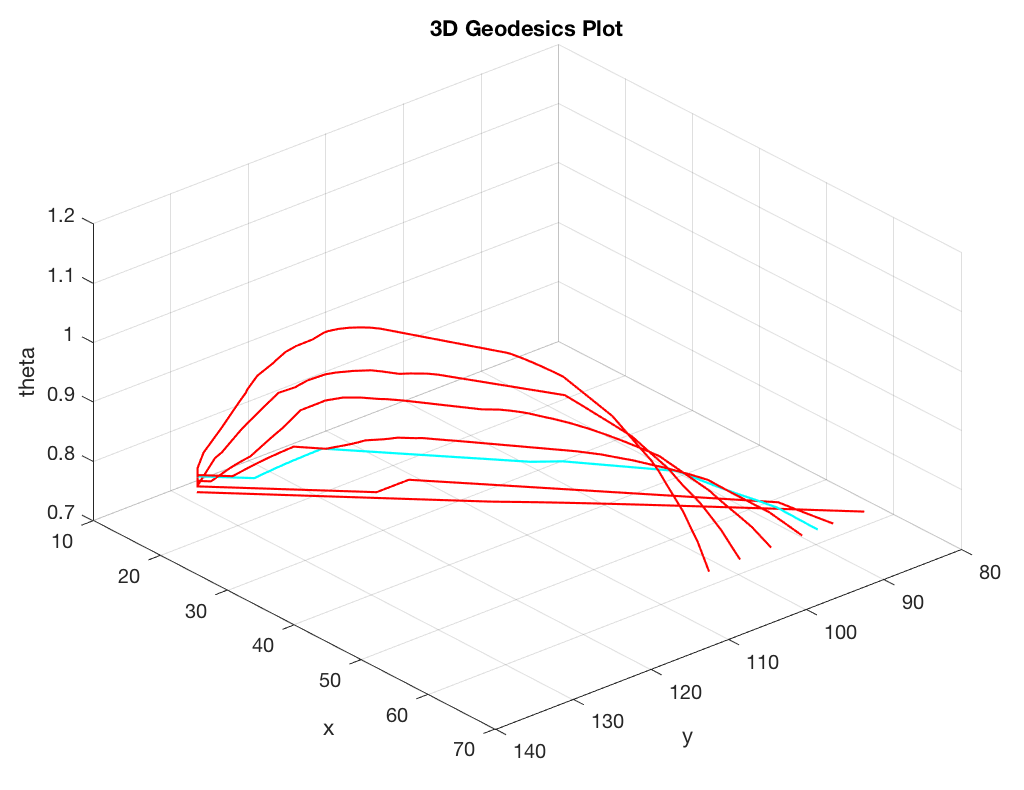} 
\caption{{\bf From left to right}: (1) minimum of distance map $d_{\epsilon, K}(\eta)$ from the boundary value condition (initial seed) of equation \eqref{eq:eikonal}, along the direction $\theta$, computed through SR-Fast-Marching. (2): 2D projection of the computed geodesics. The perceptual curve is cyan. (3): 3D plot of the computed geodesics. $\xi = 4.5$} \label{Geod_45} 
\end{figure}

\subsection{Hering illusion}\label{her:valid}
The Hering illusion, introduced by Hering in 1861~\cite{Her_1} is presented in Figure~\ref{fig:1ills}, left. In this illusion two vertical straight lines are presented in front of radial background, so that the lines appear as if they were bowed outwards. As described in the previous sections, we first convolve the distal stimulus with the entire bank of Gabor filters to compute the polarization of the metric $P(x,y,\theta)$: we take 32 orientations selected in $[0,\pi)$, $\sigma = 7.20$ pixels, $\alpha = 0.5$ pixels.  The resulting computed perceptual curves are shown in figure \ref{Geod_Her}. In order to determine the perceptual angle, we varied $\theta$ between $(0, \pi/4)$ as starting set, using $-\theta$ as ending set.

\begin{figure}[htpb]
	\centering
	\includegraphics[width=0.4\linewidth]{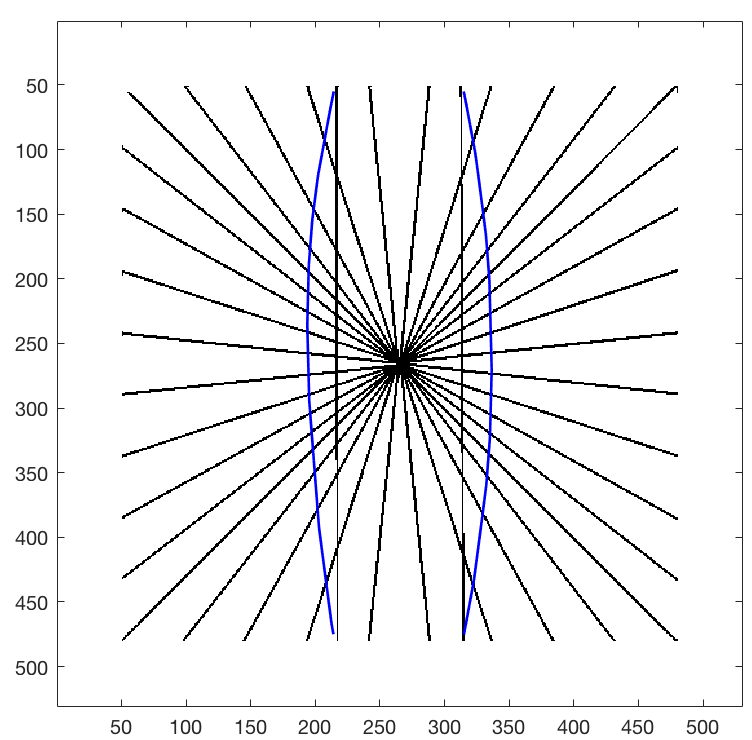} 

\caption{ Representation of the computed perceptual curves.} \label{Geod_Her} 
\end{figure}
\subsection{Ehm-Wackermann square illusion}\label{ehm:valid}
This illusion, introduced by Ehm and Wackermann in \cite{ehm2012modeling}, consists in presenting a square over a background of concentric circles, Figure~\ref{fig:1ills}, center. This context, the same we find in Ehrenstein illusion, bends the edges of the square toward the center of the image. Here we take the same number of orientations, 32, selected in $[0,\pi)$, $\sigma = 10$ pixels, $\alpha = 0.5$ pixels. In order to determine the perceptual angle, we varied $\theta$ between $(0, \pi/4)$ as starting set, using $-\theta$ as ending set. The resulting computed perceptual curves are shown in Figure~\ref{Geod_Ehm}. 

\begin{figure}[htpb]
\centering
\includegraphics[width=0.39\linewidth]{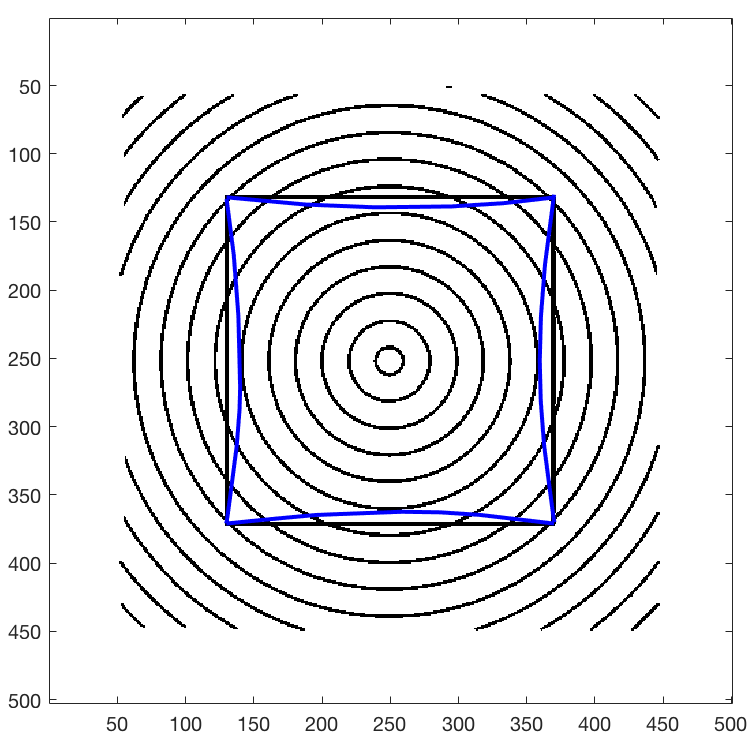} 

\caption{Representation of the computed perceptual curves.} \label{Geod_Ehm} 
\end{figure}

\subsection{Poggendorff illusion}\label{subseq:valid}
Manipulating the elements of Poggendorff to understand how to magnify the illusory phenomenon has been done in many works~\cite{day1976components,weintraub1971poggendorff}. In \cite{weintraub1971poggendorff}, the authors performed psychophysical experiments to obtain quantitative measures of the magnitude of the illusion: the illusory effects increased with increasing separation between the parallels as well as increasing the width of the obtuse angle formed by the transversal. We were not able to estimate computationally the effect induced by obtuse angles because of the interaction of the parameter $\xi$ in the sub-Riemannian Fast-Marching.

Here we consider odd Gabor filters with the following values: $\alpha = 1$, $\theta \in (-\pi, \pi)$ ($32$ values for Even Gabor profiles, 62 for Odd Gabor filters), $\sigma = 3.5$ pixels (Even) and $\sigma = 7.5$ pixels (Odd). The dimensions of images are $210 \times 102$ pixels. The scale parameter $\sigma$ is chosen in relationship with image resolution and is taken smaller for even Gabor filters, to construct filter sharp enough to detect lines. On the other hand, $\sigma$ for odd Gabor filters is taken bigger, to detect the presence of surface and obtaining vanishing integral along lines. For the entry transversal, we chose $\theta = \pi/2, \pi/4, \pi/6, \pi/10$ and width = $7,15,25$ pixels. We computed the distance between the entry trasversal and the set containing the ending points on the right side of the surface. The shortest curves computed through this model are in accord with the perceptual expectation. The angle variation of the transversal, create an increased obtuse angle effect ($\theta = \pi/10$) and a non illusory effect ($\theta = \pi/2$). In Figure~\ref{geods_poggendorff}, all the 2D projections of the computed geodesics for varying transversal entry angle and surface width is presented. When $\theta = \pi/2$, no illusion is shown and the geodesic is a straight~line.

\begin{figure}[t]
\centering
\includegraphics[width=0.17\linewidth]{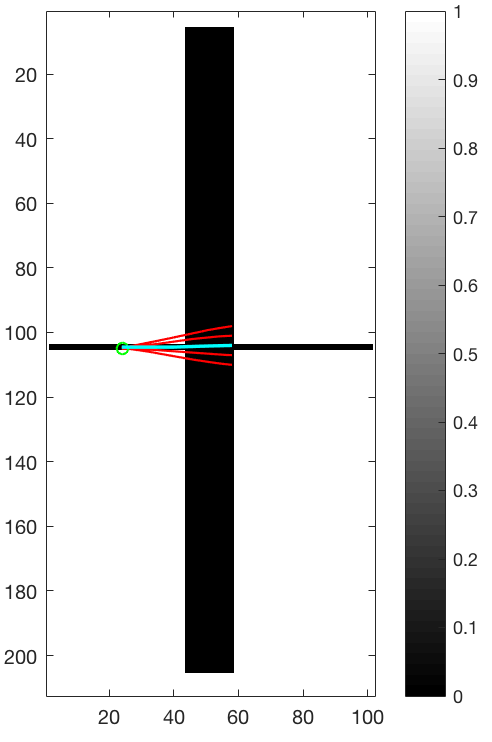} \qquad 
\includegraphics[width=0.17\linewidth]{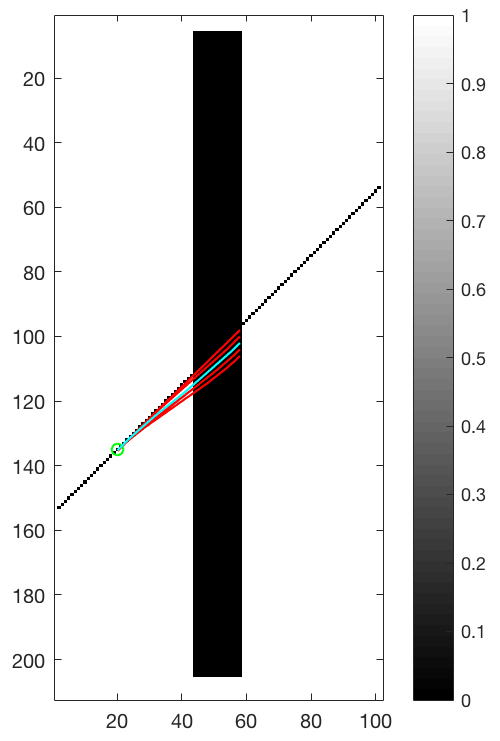} \qquad 
\includegraphics[width=0.17\linewidth]{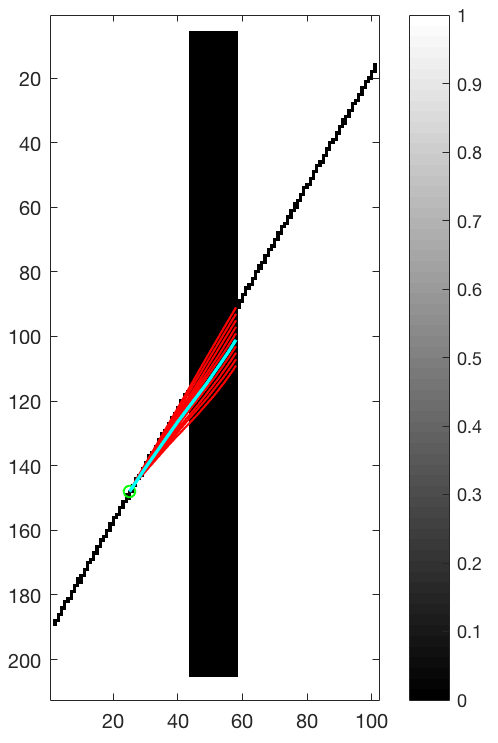} \qquad 
\includegraphics[width=0.17\linewidth]{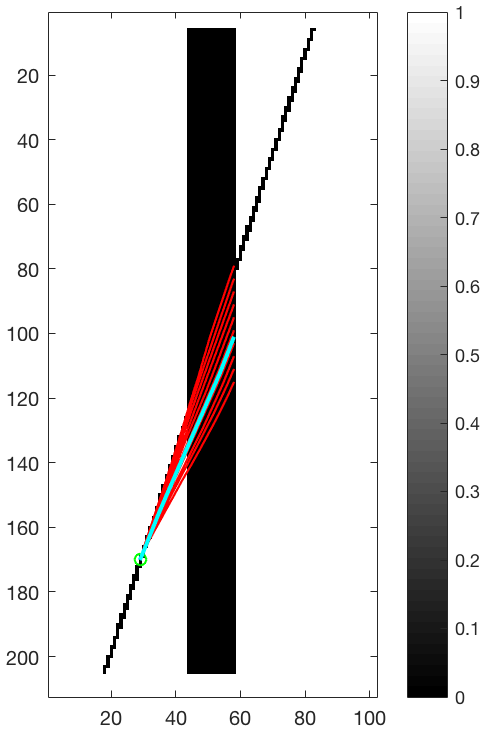} \\ \includegraphics[width=0.17\linewidth]{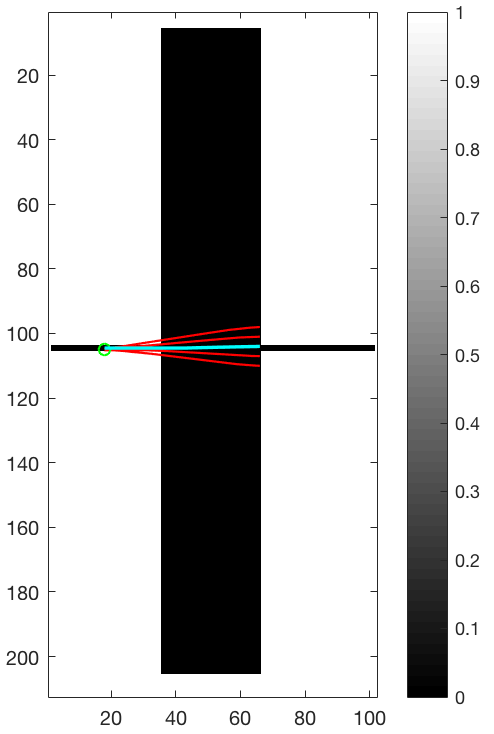} \qquad 
\includegraphics[width=0.17\linewidth]{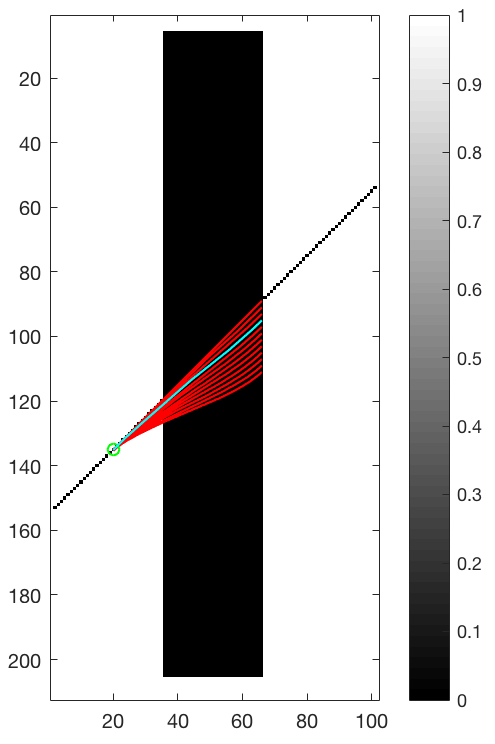} \qquad 
\includegraphics[width=0.17\linewidth]{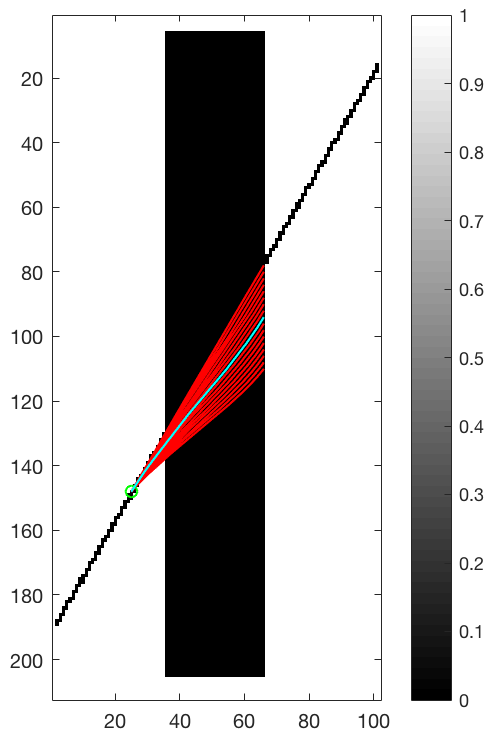} \qquad 
\includegraphics[width=0.17\linewidth]{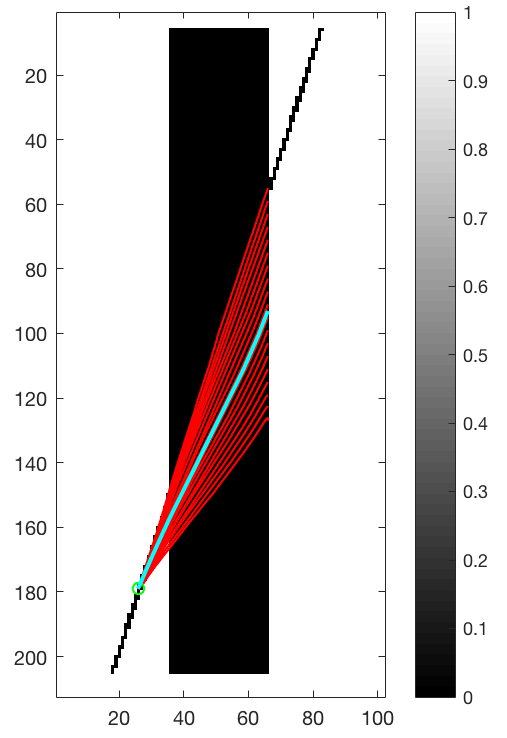} \\ 
\includegraphics[width=0.17\linewidth]{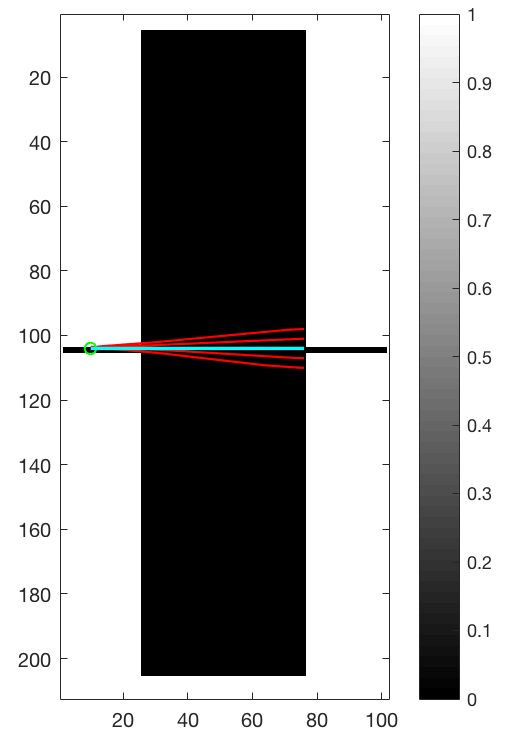}   \qquad  
\includegraphics[width=0.17\linewidth]{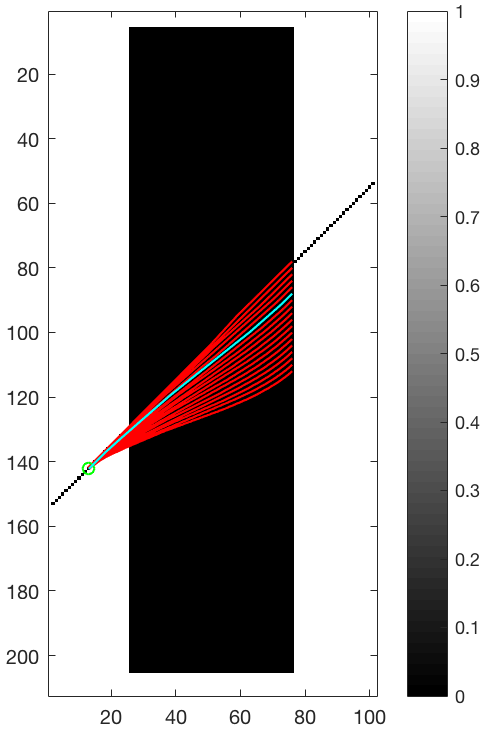} \qquad  
\includegraphics[width=0.17\linewidth]{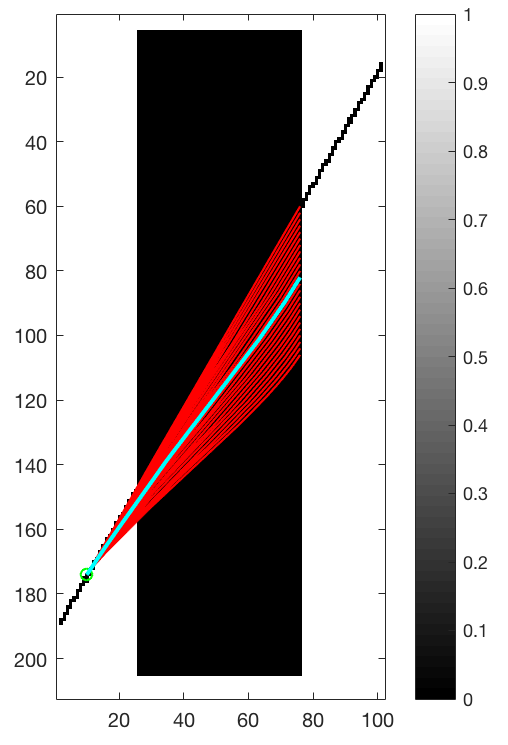} \qquad  
\includegraphics[width=0.17\linewidth]{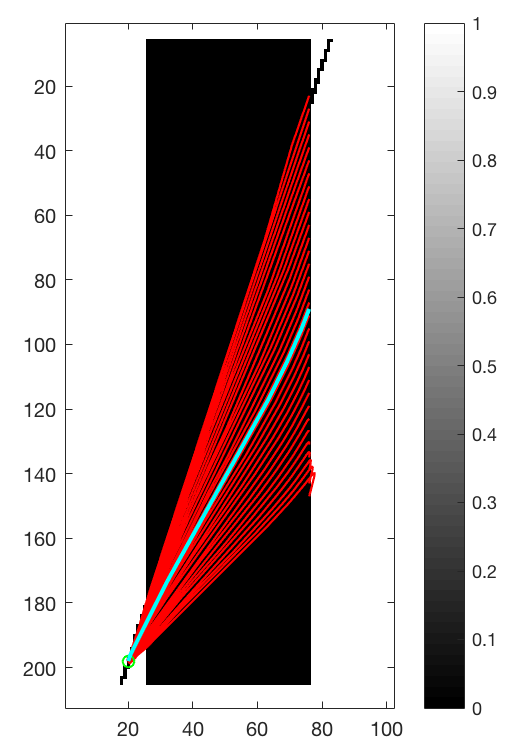}  \\

\caption{Poggendorff stimuli processed with their corresponding computed geodesics overlapped. In red we show an undersampling of geodesics computed from the left entry transversal to the right side of the central surface. Varying the orientation {\bf from Left to Right:} $\theta = \pi/2, \pi/4, \pi/6, \pi/10$ and with fixed width, {\bf Top:}  $7$ pixels, {\bf Central:}  $15$ pixels, {\bf Bottom:} $25$ pixels} \label{geods_poggendorff}
\end{figure}

\begin{figure}[htbp]
	\centering
	\includegraphics[width=0.17\linewidth]{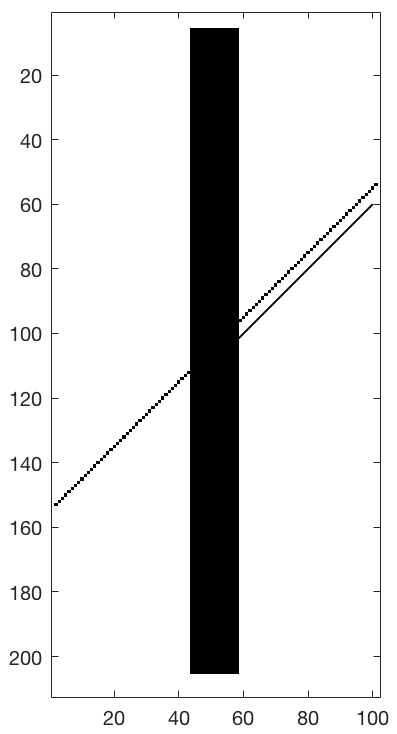} \qquad 
	\includegraphics[width=0.17\linewidth]{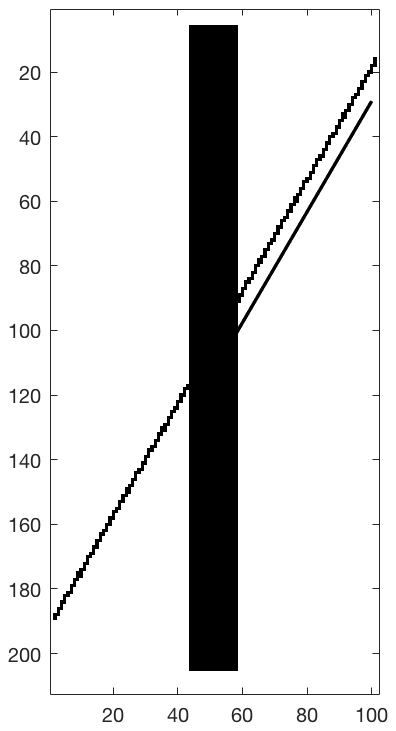} \qquad 
	\includegraphics[width=0.17\linewidth]{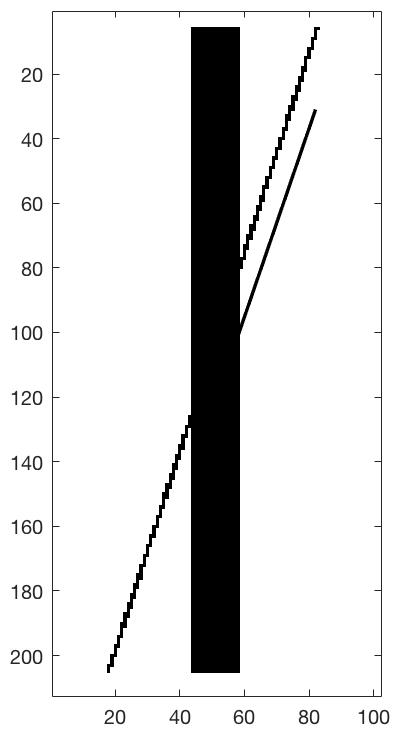} \\ 
	\includegraphics[width=0.17\linewidth]{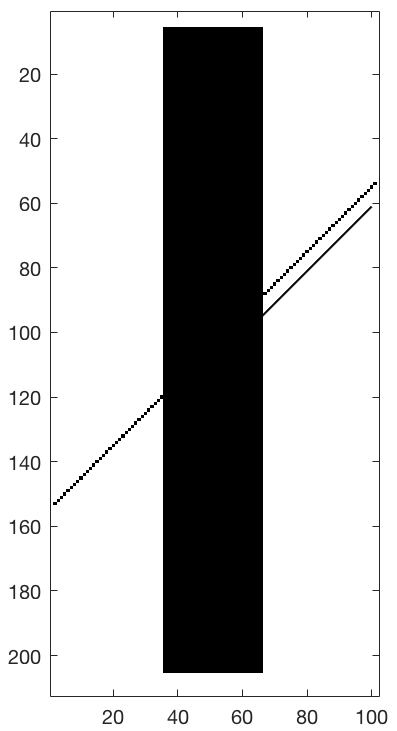} \qquad 
	\includegraphics[width=0.17\linewidth]{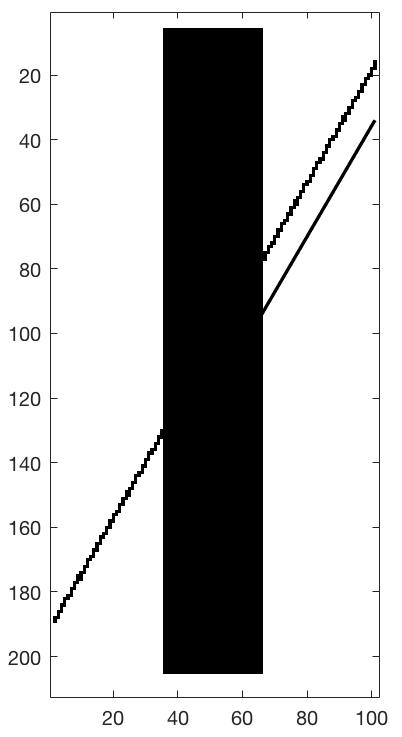} \qquad 
	\includegraphics[width=0.17\linewidth]{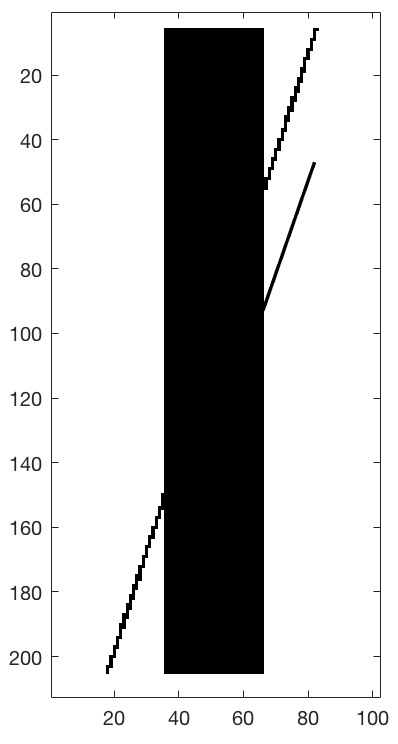} \\ 
	\includegraphics[width=0.17\linewidth]{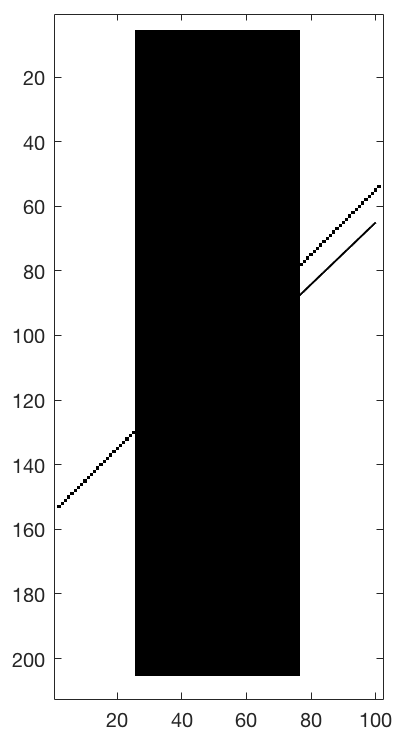} \qquad  
	\includegraphics[width=0.17\linewidth]{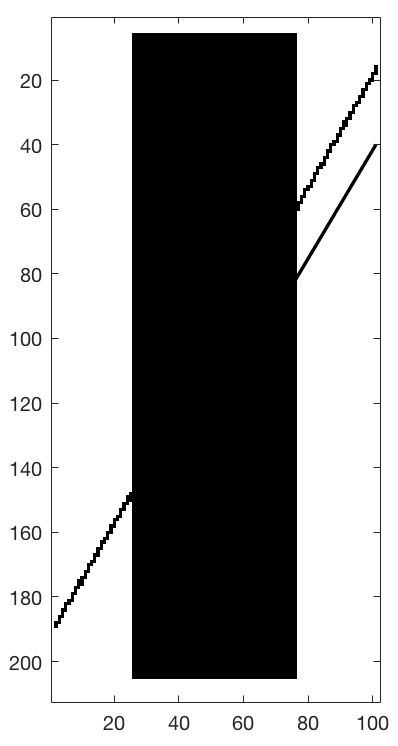} \qquad  
	\includegraphics[width=0.17\linewidth]{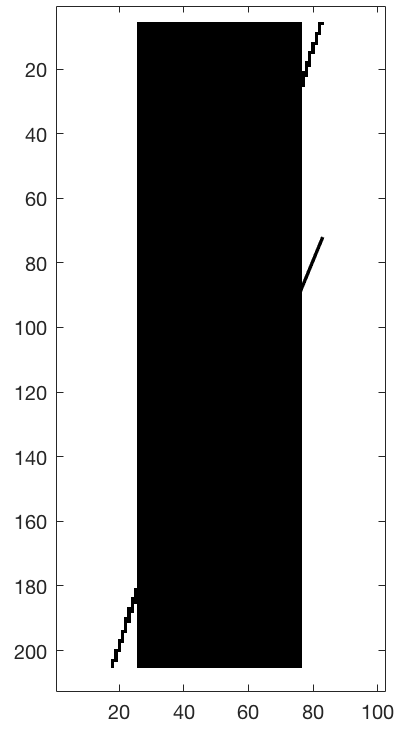} 	
	\caption{Reconstruction of the perceptual bars, from the $y$ coordinate corresponding to the cyan curve of \ref{geods_poggendorff}. Varying the orientation {\bf from Left to Right}: $\theta = \pi/4, \pi/6, \pi/10 $ for fixed width {\bf Top:}  $7$ pixels, {\bf Central:} $15$ pixels, {\bf Right:}  $25$ pixels} \label{poggendorff_reconstruction_stimuli}
\end{figure}

\subsubsection{Discussion} 

In this paragraph we show a table reporting the collected data concerning the SR lengths of the computed curve. It refers to the change of length varying the widths and angles, underlining the observed phenomena.\\
\newline
\begin{tabular}{l*{3}{|c}r}
Type of curve  & Width = 7 pixels & Width = 15 pixels & Width = 25 pixels \\
\hline
Percep. curve $\theta = \pi/4$ &  1.52 & 1.583 & 1.651 \\
Actual curve $\theta = \pi/4$ & 1.545 & 1.543 & 1.58 \\
Percep. curve $\theta = \pi/6$ & 1.007 & 1.309 & 1.173  \\
Actual curve $\theta = \pi/6$ & 1.111 & 1.173 & 1.044  \\
Percep. curve $\theta = \pi/10$ & 0.7194 & 0.9748 & 1.267  \\
Actual curve $\theta = \pi/10$ & 0.5944 & 0.7817 & 0.9503  \\
\label{tab_length_path}
\end{tabular}
\subsection{Round Poggendorff illusion}
Now we consider a variant of the Poggendorff illusion, called Round Poggendorff, see Fig.~\ref{Round_pogg}, left. The presence of the central surface induces a misperception: the arches do not seem cocircular and the left arc seems to be projected to some point with a certain orientation on the left bar. As in the previous example, we provide a terminal set to the SR-Fast Marching: the seed is fixed at the crossing point between the right arc and the right bar, $\xi = 2.5$ and possible terminal orientations belong to $[-\pi/10,0]$, where $\theta = 0$ is the angle corresponding to the orthogonal projection over the left bar and $\theta = -\pi/10$ is the boundary condition of the circle at crossing point with the left bar. The SR length of the geodesic is 1.32668 and the corresponding computed end point is $\{0.3, 0.88, -0.27\}$.

\begin{figure}[htpb]
\centering
\includegraphics[width=\textwidth]{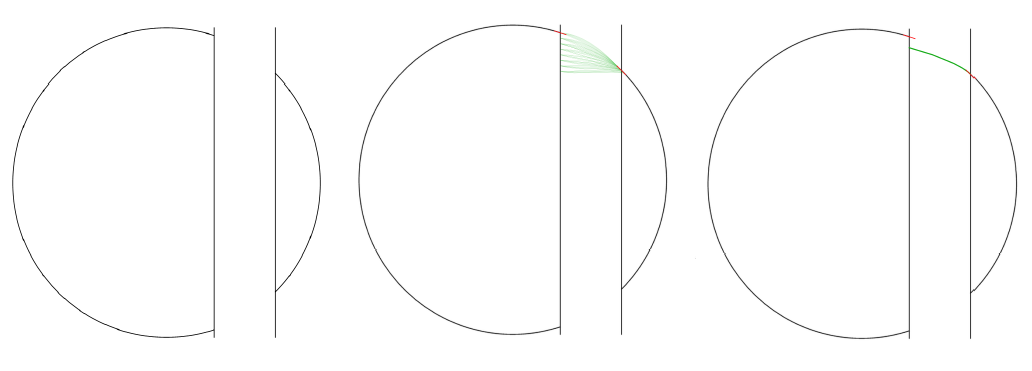}
\caption{{\bf Left:} Round Poggendorff illusion, courtesy of Talasli et al. see~\cite{talasli2015applying}. {\bf Center:} A family of geodesics starting from $(x_0,y_0,\th_0)$ with multiple endpoints. The aim is to determine $(y,\th)$ minimizing the length of the perceptual curve. {\bf Right:} A minimizer has end point $(y,\th) = (0.88, -0.27)$} \label{Round_pogg} 
\end{figure}

\section{Conclusion}
In this paper a neuro-mathematical model for the geometrical optical illusions is presented, based on the functional architecture of V1. Perceptual curves arise as geodesics of a polarized metric in $\SE$, directly induced by the visual stimulus. The geodesics are computed through SR-FM and the perceptual curves result to be shorter (w.r.t. SR-metric) than the corresponding geometrical continuation. The model has been compared with psychophysical evidences which explain how the effect varies depending on the width of the central surface and the angle of the transversal. Improving the understanding of these phenomena is very important because it can lead to insights about the behaviour of the visual cortex \cite{eagleman2001visual}, allowing new applications to be developed.   
\section*{Acknowledgment and Contributions}

The research leading to these results has received funding from the People Programme (Marie Curie Actions) of the European Union's Seventh Framework Programme FP7/2007-2013/ under REA
grant agreement n607643. 

The work of Mashtakov is supported by the Russian Science Foundation under grant 17-11-01387 and performed in Ailamazyan Program Systems Institute of Russian Academy of Sciences. 
The work on the revised manuscript was partially supported by the European Community's Seventh Framework Programme (FP7/2007-2013) / ERC grant Lie Analysis, agr. n.~335555, and  by the European Community's project  Marie Skłodowska-Curie grant agreement GHAIA n. 777822.

We thank the anonymous reviewer for the valuable comments, which helped to considerably improve the quality of the manuscript. 
 We thank A.~Ardentov for the fruitful discussions on the revised manuscript.



\bibliography{Manuscript}

\begin{thebibliography}{10}
\providecommand{\url}[1]{\texttt{#1}}
\providecommand{\urlprefix}{URL }

\bibitem{agrachev2013control}
Agrachev, A.A., Sachkov, Y.L.: Control theory from the geometric viewpoint.
  Springer-Verlag (2004)

\bibitem{Ardentov2013}
Ardentov, A.A., Sachkov, Y.L.: Conjugate points in nilpotent sub-riemannian
  problem on the engel group. Journal of Mathematical Sciences  195(3),
  369--390 (2013)

\bibitem{august2000curve}
August, J., Zucker, S.W.: The curve indicator random field: Curve organization
  via edge correlation. In: Perceptual organization for artificial vision
  systems, pp. 265--288. Springer (2000)

\bibitem{Barilari2011}
Barilari, D., Boscain, U., Gauthier, J.: On 2-step, corank 2 nilpotent
  sub-riemannian metrics. SIAM J. Control Optim.  50,  559--582 (2011)

\bibitem{barles1994solutions}
Barles, G.: Solutions de viscosit{\'e} des {\'e}quations de hamilton-jacobi,
  volume 17 of math{\'e}matiques; applications (berlin) (1994)

\bibitem{bekkers2014}
Bekkers, E., Duits, R., Berendschot, T., ter Haar~Romeny, B.: A
  multi-orientation analysis approach to retinal vessel tracking. Journal of
  Mathematical Imaging and Vision  49(3),  583--610 (2014)

\bibitem{PTR2}
Bekkers, E.J., Duits, R., Mashtakov, A., Sachkov, Y.: Vessel tracking via
  sub-riemannian geodesics on the projective line bundle. In: Nielsen, F.,
  Barbaresco, F. (eds.) Geometric Science of Information. pp. 773--781 (2017)

\bibitem{bekkers2015pde}
Bekkers, E.J., Duits, R., Mashtakov, A., Sanguinetti, G.R.: A pde approach to
  data-driven sub-riemannian geodesics in se (2). SIAM Journal on Imaging
  Sciences  8(4),  2740--2770 (2015)

\bibitem{BenYos}
Ben-Yosef, G., Ben-Shahar, O.: A tangent bundle theory for visual curve
  completion. IEEE Trans. Pattern Anal. Mach. Intell.  34,  1263--1280 (2012)

\bibitem{Boscain2008}
Boscain, U., Rossi, F.: Invariant carnot-caratheodory metrics on s3, so(3),
  sl(2) and lens spaces. SIAM J. Control Optim.  47(4),  1851--1878 (2008)

\bibitem{Brockett1982}
Brockett, R.W.: Control Theory and Singular Riemannian Geometry, pp. 11--27.
  Springer New York (1982)

\bibitem{Butt2017}
Butt, Y.A., Sachkov, Y.L., Bhatti, A.I.: Cut locus and optimal synthesis in
  sub-riemannian problem on the lie group sh(2). Journal of Dynamical and
  Control Systems  23(1),  155--195 (2017)

\bibitem{capogna2016regularity}
Capogna, L., Citti, G.: Regularity for subelliptic pde through uniform
  estimates in multi-scale geometries. Bulletin of Mathematical Sciences  6(2),
   173--230 (2016)

\bibitem{CS1}
Citti, G., Sarti, A.: A cortical based model of perceptual completion in the
  roto-translation space. Journal of Mathematical Imaging and Vision  24(3),
  307--326 (2006)

\bibitem{coren1978seeing}
Coren, S., Girgus, J.S.: Seeing is deceiving: The psychology of visual
  illusions. Lawrence Erlbaum (1978)

\bibitem{crandall1992user}
Crandall, M.G., Ishii, H., Lions, P.L.: User's guide to viscosity solutions of
  second order partial differential equations. Bulletin of the AMS  27(1),
  1--67 (1992)

\bibitem{crandall1983viscosity}
Crandall, M.G., Lions, P.L.: Viscosity solutions of hamilton-jacobi equations.
  Transactions of the American Mathematical Society  277(1),  1--42 (1983)

\bibitem{Daug}
Daugman, J.G.: Uncertainty relation for resolution in space, spatial frequency,
  and orientation optimized by two-dimensional visual cortical filters. JOSA A
  2(7),  1160--1169 (1985)

\bibitem{day1976components}
Day, R., Dickinson, R.: The components of the poggendorff illusion. British
  Journal of Psychology  67(4),  537--552 (1976)

\bibitem{deangelis1995receptive}
DeAngelis, G.C., Ohzawa, I., Freeman, R.D.: Receptive-field dynamics in the
  central visual pathways. Trends in neurosciences  18(10),  451--458 (1995)

\bibitem{duits2018}
Duits, R., Meesters, S.P.L., Mirebeau, J.M., Portegies, J.M.: Optimal paths for
  variants of the 2d and 3d reeds--shepp car with applications in image
  analysis. Journal of Mathematical Imaging and Vision  60(6),  816--848 (2018)

\bibitem{duits2007image}
Duits, R., Felsberg, M., Granlund, G., ter Haar~Romeny, B.: Image analysis and
  reconstruction using a wavelet transform constructed from a reducible
  representation of the euclidean motion group. International Journal of
  Computer Vision  72(1),  79--102 (2007)

\bibitem{duits2010left}
Duits, R., Franken, E.: Left-invariant parabolic evolutions on se (2) and
  contour enhancement via invertible orientation scores part i: Linear
  left-invariant diffusion equations on se (2). Quarterly of Applied
  Mathematics pp. 255--292 (2010)

\bibitem{duits2010left2}
Duits, R., Franken, E.: Left-invariant parabolic evolutions on se (2) and
  contour enhancement via invertible orientation scores part ii: Nonlinear
  left-invariant diffusions on invertible orientation scores. Quarterly of
  applied mathematics pp. 293--331 (2010)

\bibitem{eagleman2001visual}
Eagleman, D.M.: Visual illusions and neurobiology. Nature Reviews Neuroscience
  2(12),  920--926 (2001)

\bibitem{ehm2012modeling}
Ehm, W., Wackermann, J.: Modeling geometric--optical illusions: A variational
  approach. Journal of Mathematical Psychology  56(6),  404--416 (2012)

\bibitem{ehm2016geometric}
Ehm, W., Wackermann, J.: Geometric--optical illusions and riemannian geometry.
  Journal of Mathematical Psychology  71,  28--38 (2016)

\bibitem{Ferm}
Ferm{\"u}ller, C., Malm, H.: Uncertainty in visual processes predicts
  geometrical optical illusions. Vision research  44(7),  727--749 (2004)

\bibitem{franceschiello2017neuro}
Franceschiello, B., Sarti, A., Citti, G.: A neuro-mathematical model for
  geometrical optical illusions. Journal of Mathematical Imaging and Vision  60
  (1),  94--108 (2017)

\bibitem{ge1993collapsing}
Ge, Z.: Collapsing riemannian metrics to carnot-caratheodory metrics and
  laplacians to sub-laplacians. Can. J. Math  45(3),  537--553 (1993)

\bibitem{geisler2002illusions}
Geisler, W.S., Kersten, D.: Illusions, perception and bayes. Nature
  neuroscience  5(6),  508--510 (2002)

\bibitem{von2005treatise}
von Helmholtz, H., Southall, J.P.C.: Treatise on physiological optics, vol.~3.
  Courier Corporation (2005)

\bibitem{Her_1}
Hering, H.E.: Beitr{\"a}ge zur physiologie. 1-5, Leipzig, W. Engelmann (1861)

\bibitem{Hladky}
Hladky, R.K., Pauls, S.D., et~al.: Constant mean curvature surfaces in
  sub-riemannian geometry. Journal of Differential Geometry  79(1),  111--139
  (2008)

\bibitem{hoffman1971visual}
Hoffman, W.C.: Visual illusions of angle as an application of lie
  transformation groups. Siam Review  13(2),  169--184 (1971)

\bibitem{hoffman1989visual}
Hoffman, W.C.: The visual cortex is a contact bundle. Applied Mathematics and
  Computation  32(2),  137--167 (1989)

\bibitem{hubel1962receptive}
Hubel, D.H., Wiesel, T.N.: Receptive fields, binocular interaction and
  functional architecture in the cat's visual cortex. The Journal of physiology
   160(1),  106--154 (1962)

\bibitem{hubel1977ferrier}
Hubel, D.H., Wiesel, T.N.: Ferrier lecture: Functional architecture of macaque
  monkey visual cortex. Proceedings of the Royal Society of London B:
  Biological Sciences  198(1130),  1--59 (1977)

\bibitem{jones1987evaluation}
Jones, J.P., Palmer, L.A.: An evaluation of the two-dimensional gabor filter
  model of simple receptive fields in cat striate cortex. Journal of
  neurophysiology  58(6),  1233--1258 (1987)

\bibitem{koenderink1990receptive}
Koenderink, J.J., Van~Doorn, A.: Receptive field families. Biological
  cybernetics  63(4),  291--297 (1990)

\bibitem{lee1996image}
Lee, T.S.: Image representation using 2d gabor wavelets. IEEE Transactions on
  pattern analysis and machine intelligence  18(10),  959--971 (1996)

\bibitem{Mashtakov2017SO3}
Mashtakov, A., Duits, R., Sachkov, Y., Bekkers, E.J., Beschastnyi, I.: Tracking
  of lines in spherical images via sub-riemannian geodesics in so(3). Journal
  of Mathematical Imaging and Vision  58(2),  239--264 (2017)

\bibitem{MashtakovMTMA}
Mashtakov, A., Ardentov, A., Sachkov, Y.: Parallel algorithm and software for
  image inpainting via sub-riemannian minimizers on the group of
  rototranslations. Numer. Math. Theory Methods Appl.  6,  95--115 (2013)

\bibitem{mirebeau2014anisotropic}
Mirebeau, J.M.: Anisotropic fast-marching on cartesian grids using lattice
  basis reduction. SIAM Journal on Numerical Analysis  52(4),  1573--1599
  (2014)

\bibitem{montgomery}
Montgomery, R.: A Tour of Subriemannian Geometries, Their Geodesics and
  Applications. Mathematical Surveys and Monographs (2002)

\bibitem{monti-tesi}
Monti, R.: Distances, boundaries and surface measures in
  carnot-carath{\'e}odory spaces (2001)

\bibitem{murray2013illusory}
Murray, M.M., Herrmann, C.S.: Illusory contours: a window onto the
  neurophysiology of constructing perception. Trends in cognitive sciences
  17(9),  471--481 (2013)

\bibitem{murray2002spatiotemporal}
Murray, M.M., Wylie, G.R., Higgins, B.A., Javitt, D.C., Schroeder, C.E., Foxe,
  J.J.: The spatiotemporal dynamics of illusory contour processing: combined
  high-density electrical mapping, source analysis, and functional magnetic
  resonance imaging. The Journal of Neuroscience  22(12),  5055--5073 (2002)

\bibitem{nagel1985balls}
Nagel, A., Stein, E.M., Wainger, S.: Balls and metrics defined by vector fields
  i: Basic properties. Acta Mathematica  155(1),  103--147 (1985)

\bibitem{oppel1855uber}
Oppel, J.J.: Uber geometrisch-optische tauschungen. Jahresbericht des
  physikalischen Vereins zu Frankfurt am Main  (1855)

\bibitem{petitot2008neurogeometrie}
Petitot, J.: Neurog{\'e}om{\'e}trie de la vision. Editions Ecole Polytechnique
  (2008)

\bibitem{petitot1999vers}
Petitot, J., Tondut, Y.: Vers une neurog{\'e}om{\'e}trie. fibrations
  corticales, structures de contact et contours subjectifs modaux.
  Math{\'e}matiques informatique et sciences humaines  145,  5--102 (1999)

\bibitem{PetkovCyber}
Petkov, N., Kruizinga, P.: Computational models of visual neurons specialised
  in the detection of periodic and aperiodic oriented visual stimuli: bar and
  grating cells. Biological Cybernetics  76(2),  83--96 (1997)

\bibitem{robinson2013psychology}
Robinson, J.O.: The psychology of visual illusion. Courier Corporation (2013)

\bibitem{Bart}
Romeny, B.M.H.: Front-end vision and multi-scale image analysis: multi-scale
  computer vision theory and applications, written in mathematica. Springer
  (2008)

\bibitem{sachkovSE2fin}
Sachkov, Y.: Cut locus and optimal synthesis in the sub-riemannian problem on
  the group of motions of a plane. ESAIM C. Opt. Calc. Var.  17,  293--321
  (2011)

\bibitem{sanguinetti2015sub}
Sanguinetti, G., Bekkers, E., Duits, R., Janssen, M.H., Mashtakov, A.,
  Mirebeau, J.M.: Sub-riemannian fast marching in se (2). In: Iberoamerican
  Congress on Pattern Recognition. pp. 366--374 (2015)

\bibitem{sarti2008symplectic}
Sarti, A., Citti, G., Petitot, J.: The symplectic structure of the primary
  visual cortex. Biological Cybernetics  98(1),  33--48 (2008)

\bibitem{Sethian99a}
Sethian, J.A.: Level Set Methods and Fast {March}ing Methods. Cambridge
  University Press, second edition edn. (1999)

\bibitem{smith1978descriptive}
Smith, D.A.: A descriptive model for perception of optical illusions. Journal
  of Mathematical Psychology  17(1),  64--85 (1978)

\bibitem{song2013effective}
Song, C., Schwarzkopf, D.S., Lutti, A., Li, B., Kanai, R., Rees, G.: Effective
  connectivity within human primary visual cortex predicts interindividual
  diversity in illusory perception. The Journal of Neuroscience  33(48),
  18781--18791 (2013)

\bibitem{talasli2015applying}
Talasli, U., Inan, A.B.: Applying emmert's law to the poggendorff illusion.
  Frontiers in human neuroscience  9 (2015)

\bibitem{Vershik1988}
Vershik, A.M., Gershkovich, V.Y.: Nonholonomic problems and the theory of
  distributions. Acta Applicandae Mathematica  12(2),  181--209 (1988)

\bibitem{von1984illusory}
Von Der~Heyclt, R., Peterhans, E., Baurngartner, G.: Illusory contours and
  cortical neuron responses. Science  224 (1984)

\bibitem{wade1982art}
Wade, N.: The art and science of visual illusions. Routledge Kegan \& Paul
  (1982)

\bibitem{weintraub1971poggendorff}
Weintraub, D., Krantz, D.: The poggendorff illusion: Amputations, rotations,
  and other perturbations. Attention, Perception, Psychophysics  10(4),
  257--264 (1971)

\bibitem{weiss2002motion}
Weiss, Y., Simoncelli, E.P., Adelson, E.H.: Motion illusions as optimal
  percepts. Nature neuroscience  5(6),  598--604 (2002)

\bibitem{westheimer2008illusions}
Westheimer, G.: Illusions in the spatial sense of the eye:\!
  geometrical--\-optical\! illusions and the neural\! representation of space.
  Vision research  48(20),  2128--2142 (2008)

\bibitem{young1987gaussian}
Young, R.A.: The gaussian derivative model for spatial vision: I. retinal
  mechanisms. Spatial vision  2(4),  273--293 (1987)

\end{thebibliography}
\bibliographystyle{ieeetr}

\end{document}